\documentclass[10pt,final,twocolumn]{IEEEtran}

\usepackage[latin1]{inputenc}
\usepackage[english]{babel}
\usepackage{amsmath,amssymb}
\usepackage{graphicx,graphicx,verbatim}
\usepackage{verbatim}
\usepackage{color}

\newtheorem{thm}{Theorem}
\newtheorem{conj}{Conjecture}
\newtheorem{cor}{Corollary}
\newtheorem{lemma}{Lemma}

\newtheorem{assumption}{Assumption}

\newtheorem{defi}{Definition}
\newtheorem{prop}{Proposition}

\newtheorem{exx}{Example}
\newtheorem{remm}{Remark}

\newenvironment{corollary}{\begin{cor}\rm }{\hfill \hspace*{1pt} \hfill $\lrcorner$ \end{cor}}

\newenvironment{theorem}{\begin{thm}\rm }{\hfill \hspace*{1pt} \hfill $\lrcorner$ \end{thm}}
\newenvironment{definition}{\begin{defi}\rm }{\hfill \hspace*{1pt} \hfill $\lrcorner$ \end{defi}}
\newenvironment{remark}{\begin{remm}\rm }{\hfill \hspace*{1pt} \hfill $\lrcorner$\end{remm}}
\newenvironment{example}{\begin{exx}\rm }{\hfill \hspace*{1pt} \hfill $\lrcorner$ \end{exx}}
\newenvironment{proofof}{\noindent {\em Proof of }}{\hfill \hspace*{1pt}
\hfill $\blacksquare$}
\newenvironment{proof}{\noindent {\em Proof.}}{\hfill \hspace*{1pt} \hfill $\square$}
\newenvironment{conjecture}{\begin{conj}\rm }{\hfill \hspace*{1pt} \hfill $\lrcorner$ \end{conj}}

\newcommand\real{\ensuremath{{\mathbb R}}}
\newcommand\realn{\ensuremath{{\mathbb{R}^n}}}

\newcommand{\dom}{\mathrm{dom}\,}
\newcommand\mymatrix[2]{\left[\begin{array}{#1} #2 \end{array}\right]}

\newcommand{\calA}{\mathcal{A}}
\newcommand{\calB}{\mathcal{B}}
\newcommand{\calC}{\mathcal{C}}

\newcommand{\calE}{\mathcal{E}}

\newcommand{\calR}{\mathcal{R}}
\newcommand{\calS}{\mathcal{S}}
\newcommand{\calU}{\mathcal{U}}
\newcommand{\calV}{\mathcal{V}}

\newcommand{\calK}{\mathcal{K}}
\newcommand{\calW}{\mathcal{W}}
\newcommand{\calY}{\mathcal{Y}}
\newcommand{\calX}{\mathcal{X}}
\newcommand{\calZ}{\mathcal{Z}}

\begin{document}

\title{Differentially positive systems
\thanks{
F. Forni and R. Sepulchre are with the University of Cambridge, Department of Engineering, 
Trumpington Street, Cambridge CB2 1PZ, and with the Department of Electrical Engineering and Computer Science, 
University of Li{\`e}ge, 4000 Li{\`e}ge, Belgium, \texttt{ff286@cam.ac.uk|r.sepulchre@eng.cam.ac.uk}.
The research is supported by FNRS.
The paper presents research results of the Belgian Network DYSCO
(Dynamical Systems, Control, and Optimization), funded by the
Interuniversity Attraction Poles Programme, initiated by the Belgian
State, Science Policy Office. The scientific responsibility rests with
its authors. }}
\author{F. Forni, R. Sepulchre}
\date{\today}

\maketitle

\begin{abstract}  % Abstract of not more than 200 words.
The paper introduces and studies differentially positive systems, that is, systems whose linearization along an arbitrary trajectory is positive. A generalization of Perron Frobenius theory is developed in this differential framework to show that the property induces a (conal) order that strongly constrains the asymptotic behavior of solutions. The results illustrate that behaviors constrained by local order properties extend beyond the well-studied class of linear positive systems and monotone systems, which both require a constant cone field and a linear state space. 
\end{abstract}

\section{Introduction}

Positive systems are linear behaviors that leave a cone invariant \cite{Bushell1973}. They have a rich history both because of the relevance of the property in applications (e.g., when modeling a behavior with positive variables \cite{Luenberger1979,Farina2000,Hirsch2003}) and because the property significantly restricts the behavior, as established by Perron Frobenius theory: if the cone invariance is strict, that is, if the boundary of the cone is eventually mapped to the interior of the cone, then the asymptotic behavior of the system lies on a one dimensional object. Positive systems find many applications in systems and control, ranging from specific stabilization properties \cite{Willems1976, Muratori1991,Farina2000,DeLeenheer2001,Knorn2009, Roszak2009}
to observer  design \cite{Hardin2007,Bonnabel2011}, and to 
distributed control \cite{Moreau2004, Rantzer2012, Sepulchre2010}. 

Motivated by the importance of positivity in linear systems theory, the present paper investigates the behavior of differentially positive systems, that is, systems whose linearization along trajectories is positive. We discuss both the relevance of the property for applications and how much the property restricts the behavior, by generalizing Perron Frobenius theory to the differential framework. The conceptual picture is that a cone is attached to every point of the
state space, defining a cone field, and that contraction of that cone field along the flow eventually constrains the behavior to be one-dimensional.

Differential positivity reduces to the well-studied property of monotonicity when the state-space is a linear vector space and when the cone field is constant. First studied for closed systems \cite{Smith1995,Hirsch1995,Hirsch1988,Dancer1998} and later extended to open systems 
\cite{Angeli2003,Angeli2004a,Angeli2004b}, the concept of monotone systems encompasses cooperative and competitive systems \cite{Hirsch2003,Piccardi2002} and is extensively adopted in biology and chemistry for modeling and control purposes \cite{DeLeenheer2004,Enciso2005, DeLeenheer2007,Sontag2007,Angeli2008,Angeli2012}.
Differential positivity is an infinitesimal characterization of monotonicity. The differential viewpoint allows for a generalization of monotonicity because the state-space needs not be linear and the cone needs not be constant in space. The generalization is relevant in a number of applications. In particular, non-constant cone fields in linear spaces and invariant cone fields on nonlinear spaces are two situations frequently encountered in applications. Like monotonicity, differential positivity induces an order between solutions. But in contrast to monotone systems, the \emph{conal} order needs not to induce a partial order \emph{globally}, allowing for instance to (locally) order solutions on closed curves, such as along limit cycles or in nonlinear spaces such as the circle.

A main contribution of the paper is to generalize Perron-Frobenius theory in the differential framework. The Perron-Frobenius vector of linear positive systems here becomes a vector field and the integral curves of the Perron-Frobenius vector field shape the attractors of the system. A main result of the paper is to provide a characterization of 
limit sets of differentially positive systems akin to Poincar{\'e}-Bendixson
theorem for planar systems. Differentially positive systems can model multistable behaviors, excitable behaviors, oscillatory behaviors, but preclude for instance the existence of attractive homoclinic orbits, and a fortiori of strange attractors. In that sense, differentially positive systems single out a significant class of nonlinear systems that have a simple asymptotic behavior.

The paper is organized as follows. Section \ref{sec:nutshell} introduces the main ideas of differential positivity on familiar phase portraits and at an intuitive level. It aims at showing that the differential concept of positivity is a natural one. Section \ref{sec:notation} covers some mathematical preliminaries and notations while Section \ref{sec:conal_manifold} summarizes the main mathematical notions of order on manifolds. The next three sections contain the main results of the paper: the formal notion of differentially positive system, differential Perron-Frobenius theory, and a characterization of limit sets of differentially positive systems. 
Section \ref{sec:pendulum} illustrates several important points of the paper on the popular nonlinear pendulum example. 
{Proofs are in appendix}. Our treatment of differential positivity is for continuous-time and discrete-time open systems. The important topic of interconnections of differentially positive systems is a rich one and will be discussed in a separate paper.

\section{Differential positivity in a nutshell}
\label{sec:nutshell}

Figure \ref{fig:phase_portraits} illustrates four different phase portraits of (closed) differentially positive systems. Two of the phase portraits are represented in two different set of coordinates. The figure illustrates that for each of the phase portraits, a cone can be attached at any point in such a way that the cone is infinitesimally contracted by the flow
{(i.e. the cone angle shrinks under the action of the flow)}.
Furthermore, the cones can be patched to each other to define a smooth cone field.

\begin{figure}[tbp]
\centering
\includegraphics[width=1\columnwidth]{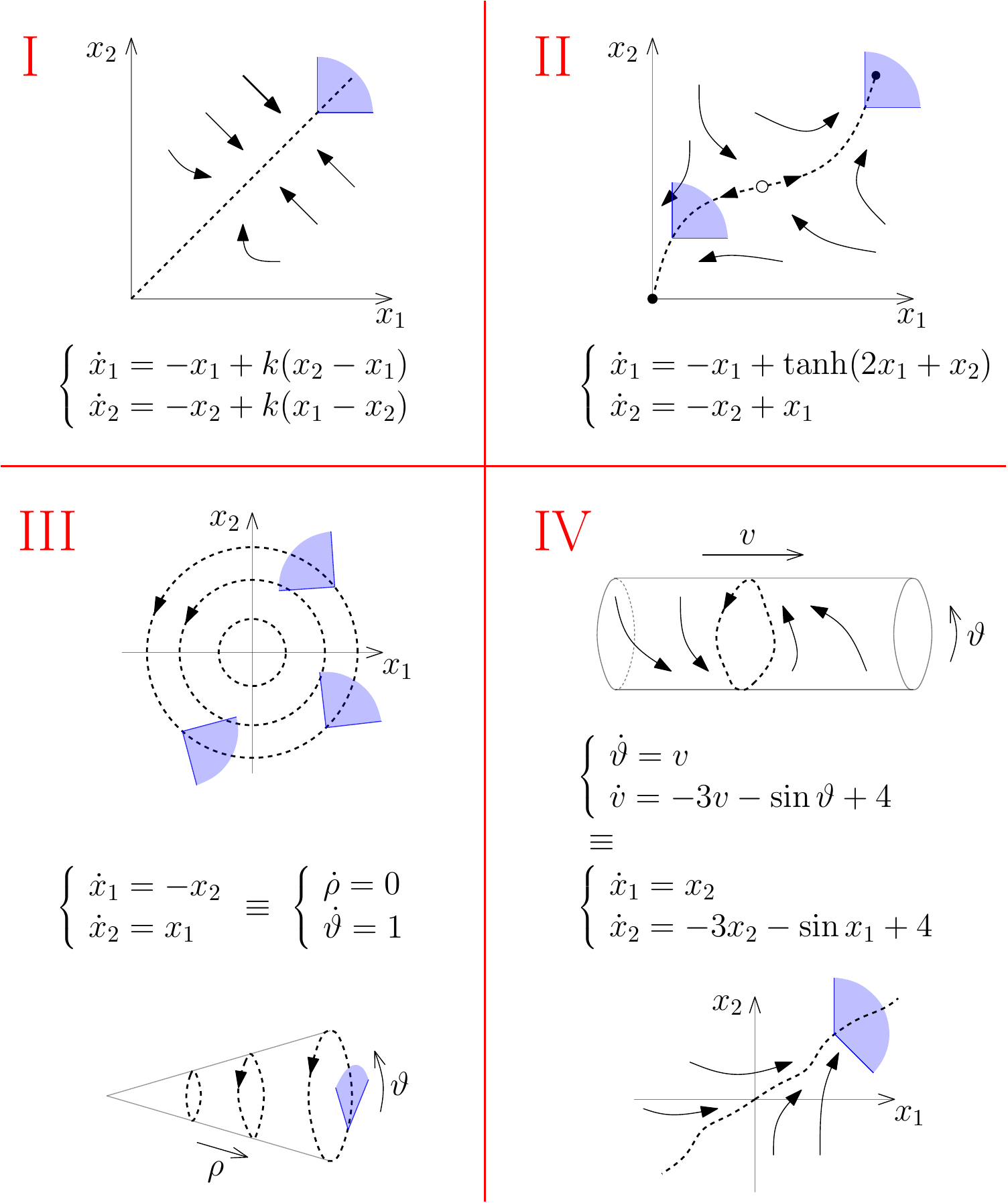}  
\caption{The phase portraits of four different planar differentially positive systems.
(I) a linear consensus model. 
(II) a monotone bistable model.
(III) the harmonic oscillator. 
(IV) the nonlinear pendulum with strong damping and torque.
}
\label{fig:phase_portraits}
\end{figure}

In the first two examples, the cone is actually the same everywhere, defining a \emph{constant} cone field in a \emph{linear} space. In the third example, both the state-space and the dynamics are linear but the cone rotates with the flow. It defines a \emph{non constant} cone field in a \emph{linear} space. In the fourth example, the cone field must be defined infinitesimally because the state space is nonlinear. At each point, a cone is defined in the tangent space. The \emph{nonlinear} cylindrical space $\mathbb{S}\times\real$ 
is a Lie group and the cone is moved from point to point by (left) translation. The analogy between the third and fourth examples is apparent when studying the phase portrait of the harmonic oscillator in polar coordinates. The nonlinear change of coordinates makes the cone invariant on the conic nonlinear space $\mathbb{R}_+\times\mathbb{S}$. The analogy between the first, second, and fourth examples is apparent when unwrapping the phase portrait of the nonlinear pendulum in the plane. In cartesian coordinates, that is, unwrapping the angular coordinate $\vartheta$ on the real line, the cone field becomes constant in a linear space.

The first phase portrait is the phase portrait of a linear system that leaves the positive orthant invariant. It is a strictly positive system. Its behavior is representative of consensus behaviors extensively studied in the recent years \cite{Moreau2004,Olfati-Saber2007,Sepulchre2010a}. The second phase portrait leaves the same cone invariant but the dynamics are nonlinear. Here the cone invariance can be characterized differentially: the linearization along any trajectory is a positive linear system with respect to the positive orthant. It is an example of monotone system, representative of bistable behaviors extensively studied in decision-making processes, see e.g. \cite{Trotta2012}. The third example is the phase portrait of the harmonic oscillator. Solutions cannot be globally ordered in the state space because the trajectories are closed curves. But the positivity of the linearization is nevertheless apparent in polar coordinates. The corresponding order property will be characterized by the notion of conal order on manifolds developed in Section IV. The fourth example is the phase portrait of the nonlinear pendulum with strong damping. Positivity of the linearization and differential positivity of the nonlinear pendulum is studied in details in the last Section of the paper.

The main message of the paper is that differential positivity constrains the asymptotic behavior of the four different phase portraits in a similar way. For linear positive systems, this is Perron- Frobenius theory. The Perron-Frobenius vector attracts all solutions to a one-dimensional ray. For differentially positive systems, the generalized object is a Perron-Frobenius curve, an integral curve of the Perron-Frobenius vector field characterized in Section V. In the second phase portrait, this is the heteroclinic orbit connecting the two stable equilibria and the unstable saddle equilibrium. In the third phase portrait, every trajectory is a Perron-Frobenius curve. The differential positivity is not strict in that case. In the fourth phase portrait, all solutions
except the unstable equilibrium are attracted to a single Perron-Frobenius curve, the limit cycle. 

The convergence properties of differentially positive systems are a consequence of the infinitesimal contraction of cones along trajectories. The significance of the property is that it can be checked locally but that it discriminates among different types of global behaviors. The smoothness of the cone field is what connects the local property to the global property. A most important feature of differential positivity is that it allows saddle points such as in Figure \ref{fig:phase_portraits}\hspace{-0.5mm}.II, because the local order is compatible with a global smooth cone field, but that it does not allow saddle points such as in Figure \ref{fig:homoclinic}. The homoclinic orbit makes the local order dictated by the saddle point incompatible with a global smooth cone field. This incompatibility has been recognized since the early work of Poincar{\'e} as the essence of complex behaviors. 
In contrast, the limit sets
of differentially positive systems are simple, 
in a sense that is made precise in 
Section \ref{sec:limit_sets}.

\begin{figure}[b]
\centering
\includegraphics[width=0.55\columnwidth]{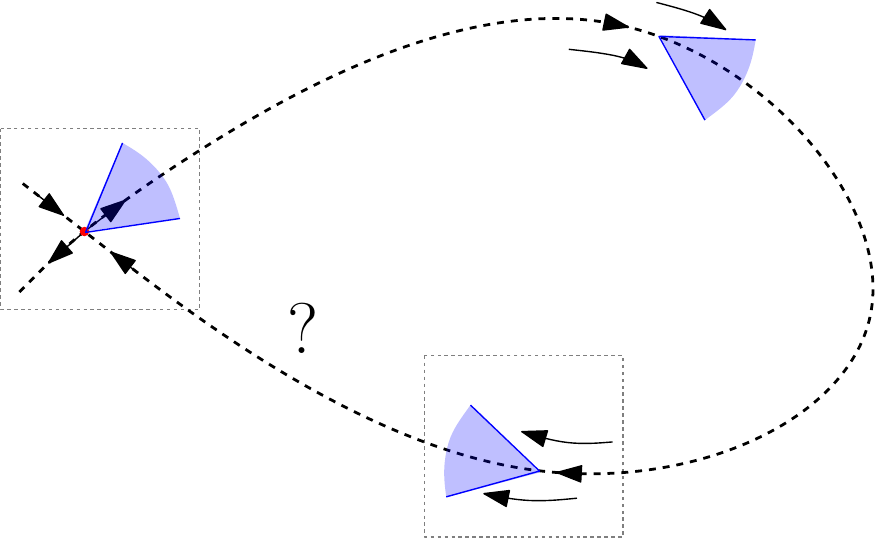}  
\caption{Strict differential positivity excludes homoclinic orbits on saddle points whose unstable manifold has dimension one.
The local order imposed by the saddle point cannot be extended 
globally to a smooth cone field.}
\label{fig:homoclinic}
\end{figure}

\section{Mathematical preliminaries  \\ and basic assumptions}
\label{sec:notation}

A ($d$-dimensional) manifold $\calX$ is a couple $(\calX,\mathcal{A}^+)$
where $\calX$ is a set and $\mathcal{A}^+$ is a maximal atlas of $\calX$
into $\real^d$, such that the topology induced by $\mathcal{A}^+$ is Hausdorff and 
second-countable {(we refer to this topology as the manifold topology).
Throughout the paper every manifold is connected.} 
$T_x\calX$ denotes the tangent space at $x$ and 
$T\calX:= \bigcup_{x\in\calX} \{x\}\times T_x\calX$ denotes the tangent bundle.
$\calX$ is endowed with a Riemannian metric tensor, represented by
a (smoothly varying) inner product $\langle \cdot,\cdot\rangle_x\,$.
$|\delta x|_x := \sqrt{\langle \delta x, \delta x \rangle_x}$, for any $\delta x\in T_x\calX$.
{The Riemannian metric endows the manifold with the 
Riemannian distance $D$. We assume that $(\calX,D)$ is a complete 
metric space (see e.g. \cite[Section 3.6]{Absil2008}).
The metric space topology and the manifold topology agree
\cite[Theorem 3.1]{boothby2003}.}

Given two smooth manifolds $\calX_1$, $\calX_2$,
the differential  of $f:\calX_1\to\calX_2$
at $x$ is denoted by
$\partial f(x) :T_x \calX \to T_{f(x)}\calX$.
Given $\calX_1 =\calZ\times\calY$, {the operator
$\partial_z f(z,y) $ satisfies} $\partial_z f(z,y) \delta z := \partial f(z,y) [\delta z,0]$ 
for each $(z,y)\in\calX_1$ and each $[\delta z, 0] \in T_{(z,y)} \calX_1$, where $\delta z \in T_z\calZ$.
Finally, we write $f(\cdot,y)$ to denote the function in $\calZ \to \calX_2$ 
mapping each $z\in \calZ$ into $f(z,y)\in \calX_2$. 

A \emph{curve} or $\gamma$ on $\calX$, 
{
is a mapping $\gamma :I \to \calX$ where 
either $I \subseteq \real$ or $I \subseteq \mathbb{Z}$}.  
$\mathrm{dom}\gamma$ and $\mathrm{im}\gamma$ denote domain and image of $\gamma$.
We say that a curve $\gamma:I\to\calX$ is bounded
if $\mathrm{im} \gamma$ is a bounded set. 
We sometime use
$\dot{\gamma}(s)$ or $\frac{d\gamma(s)}{ds}$ to denote $\partial \gamma(s)1$,
for $s \in \mathrm{dom} \gamma$. 

Given a set $\mathcal{S}\subseteq\calX$, 
$\mathrm{int}\mathcal{S}$ and $\mathrm{bd}\calS$ denote 
{interior} and {boundary} of $\calS$, respectively. 
Given a vector space $\calV$, a set $\mathcal{S}\subseteq\calV$,
and a constant $\lambda \in \real$, 
$\lambda \calS$ denotes the set $\{\lambda x\in \mathcal{V}\,|\, x \in \calS\}$.
$\calS+\calS$ denotes the set $\{x+y\in \mathcal{V}\,|\, x,y \in \calS\}$.
{ Given a point $y\in \calS$, 
$\calS\!\setminus\!\!\{y\} := \{x\in \calS\,|\, x\neq p\}$.
Given a sequence of sets $\calS_n$, 
$ \lim_{n\to\infty} \calS_n$ is the usual set-theoretic limit based on the 
Painlev{\'e}-Kuratowski convergence \cite[Chapter 4]{Rockafellar2004}.}

Let $\Sigma$ be an open continuous dynamical system with
{
(smooth) state manifold $\calX$ and input manifold $\calU$
},
represented by
$ \dot{x} = f(x,u)$, $(x,u)\in\calX\times\calU $,
where $f$ is a (input-dependent) vector field that assigns to each 
$(x,u) \in \calX\times\calU$ a tangent vector $f(x,u) \in T_x \calX$.
We make the \emph{standing assumptions} that
the vector field $f$ and $u(\cdot)$ are $C^2$ functions.
Following \cite[Chapter 4, Section 4]{boothby2003}, 
two differentiable curves $x(\cdot):I\subseteq\real\to\calX$ (trajectory) and
$u(\cdot):I\subseteq\real\to\calU$ (input)
are a solution pair $(x(\cdot),u(\cdot)) \in \Sigma$ 
if $\dot{x}(t) = f(x(t),u(t))$ for all $t\in I$. 
An open discrete dynamical system $\Sigma$ is
represented by the recursive equation
$x^+ = f(x,u)$, $(x,u)\in \calX \times\calU$.
We make the \emph{standing assumption} 
that { $f:\calX\times\calU \to \calX$} is a $C^1$ function.
$(x(\cdot),u(\cdot)) : [t_0,\infty)\subseteq\mathbb{Z} \to \calX\times\calU$ is 
a solution pair of $\Sigma$ if $x(\cdot)$ and $u(\cdot)$ satisfy 
$x(t+1) = f(x(t),u(t))$ for each $t\in [t_0,\infty)$.

In what follows, we make the simplifying assumption of
\emph{forward completeness} of the solution space, namely
that every solution pair has domain $I = [t_0,\infty) \subseteq \real$ ($\subseteq \mathbb{Z}$). 
Given the solution pair 
$(x(\cdot),u(\cdot)) : [t_0,\infty) \to \calX \times \calU \in \Sigma$
we say that $\psi(\cdot,t,x(t),u(\cdot)):= x(\cdot)$ is the \emph{trajectory} or
the \emph{integral curve} passing through 
$x(t)$ at time $t\geq t_0$ under the action of the input $u(\cdot)$. 
For constant inputs we simply write $\psi(\cdot,t,x(t),u)$
and for closed systems we use $\psi(\cdot,t,x(t))$. 
The \emph{flow} of $\Sigma$ is given by the quantity $\psi(t,t_0,\cdot,u)$ for any $t\geq t_0$.
For any curve $\gamma(\cdot)$ and set $\calS$,
$\psi(t,t_0,\gamma(\cdot),u)$ denotes the time evolution of 
$\gamma(\cdot)$ along the flow of the system at time $t$,
and $\psi(t,t_0,\calS,u)$ denotes the set $\{\psi(t,t_0,x,u) \,|\, x\in\calS\}$.
For closed systems
we say that $x_\omega\in\calX$ is an $\omega$-limit point of a trajectory $x(\cdot)$ 
if there exists a sequence of times $t_k \to \infty$ as $k \to \infty$ such that
$x_\omega = \lim_{k\to\infty} x(t_k)$. In a similar way, an
$\alpha$-limit point of a trajectory $x(\cdot)$ is given by 
$\lim_{k\to\infty} x(t_k)$ for some sequence 
$t_k \to -\infty$ as $k \to \infty$.
The $\omega$-limit set 
$\omega(x_0)$ ($\alpha$-limit set) 
is the union of the $\omega$-limit points ($\alpha$-limit points)
of the trajectory $x(\cdot)$ from the initial condition $x(t_0)=x_0$.

\section{Cone fields, conal curves, conal orders}
\label{sec:conal_manifold}
A \emph{conal manifold} $\calX$ is a smooth manifold endowed with a \emph{cone field} 
\cite{Lawson1989},
\begin{equation}
\calK_\calX(x)\subseteq T_{x}\calX \qquad\quad \forall x \in \calX \ .
\end{equation}
Like for vector fields, a cone field attaches to each point $x$ of the manifold a cone $\calK_\calX(x)$
defined in the tangent space $T_x\calX$. 
\emph{Throughout the paper}, each cone
$\calK_\calX(x)\subseteq T_x\calX$ is closed, pointed and convex (for each $x \in \calX$,
$\calK_\calX(x) + \calK_\calX(x) \subseteq \calK_\calX(x)$, 
$\lambda \calK_\calX(x) \subseteq \calK_\calX(x)$ for any $ \lambda\in\real_{\geq 0}$,
and $ \calK_\calX(x) \cap - \calK_\calX(x) = \{0\}$).
To avoid pathological cases, we assume that each cone is solid 
(i.e. {it contains $n$ independent tangent vectors, 
where $n$} is the dimension of the tangent space) and
there exists a linear invertible mapping  
$\Gamma(x_1,x_2):T_{x_1}\calX \to T_{x_2}\calX$ 
for each  $x_1, x_2\in\cal\calX$, 
such that 
$\Gamma(x_1,x_2)\calK_\calX(x_1) = \calK_\calX(x_2)$.
{

Note that the application of a linear invertible mapping to a cone is 
intended as an operation 
on the rays of the cone, that is,
$
\Gamma(x_1,x_2) \calK_\calX(x_1) := 
\{\lambda \Gamma(x_1,x_2)\delta x \in T_{x_2}\calX \, | \, \delta  x \in \calK_\calX(x_1), \lambda \geq 0 \}
$.
}

We make the \emph{standing assumption} that each cone field is smooth.
In particular, { in local coordinates},
\begin{equation}
\label{eq:basic_cone_field}
\calK_\calX(x) = \{
 \delta x \in T_{x}\calX 
 \ | \
 \forall i\in I , \
 k_i(x,\delta x) \geq 0 
\}
\end{equation}
where $I\subseteq\mathbb{Z}$ is an index set and $k_i:T\calX \to \real$ are functions;
and we say that a cone field is smooth if the functions $k_i$ are smooth.

A curve $\gamma:I\subseteq\real\to\calX$
is a \emph{conal curve} on $\calX$ if 
\begin{equation}
\label{eq:conal_curve}
\dot{\gamma}(s) \in \calK_\calX(\gamma(s)) \qquad\qquad \mbox{ for all } s\in I \ . 
\end{equation}
Conal curves are integral curves of the cone field, as shown
in Figure \ref{fig:conal_order}.  They endow
the manifold with a local partial order:
for each $x_1,x_2\in\calX$, $x_1 \sqsubseteq_{\calK_\calX} x_2$
if and only if there exists a conal curve $\gamma:I\subseteq\real \to \calX$ 
such that
$ \gamma(s_1) = x_1$ and $\gamma(s_2) = x_2 $ for some $s_1 \leq s_2$. 
\begin{figure}[b]
\centering
\includegraphics[width=0.5\columnwidth]{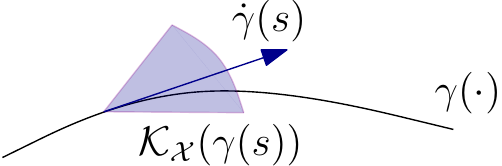}  
\caption{A conal curve satisfies \eqref{eq:conal_curve}. 
If $s_1 \leq s_2$ then $\gamma(s_1) \sqsubseteq_{\calK_\calX} \gamma(s_2)$.
}
\label{fig:conal_order}
\end{figure}

The \emph{conal order} $\sqsubseteq_{\calK_\calX}$ is the natural 
generalization on manifolds of the notion of partial order on vector spaces.
In fact $\sqsubseteq_{\calK_\calX}$ is a partial order 
when $\calX$ is a vector space and the cone field
$\calK_\calX(x) = \calK_\calX$ is constant:
two points $x,y \in \calX$ satisfy
$x\sqsubseteq_{\calK_\calX} y$ iff $y-x \in \calK_\calX$, 
as shown in \cite[Proposition 1.10]{Lawson1989},
which is the usual definition of a
partial order on vector spaces \cite[Chapter 5]{Schaefer1971}.
In general, $\sqsubseteq_{\calK_\calX}$ is not a (global) partial order on $\calX$ since
antisymmetry may fail.
The reader is referred to \cite{Lawson1989} and \cite{Neeb1993} 
for a detailed exposition of the relations among cone fields, 
ordered manifolds, and homogeneous spaces. 

\begin{example}
For the manifold $\mathbb{S}\times\real$ in Figure \ref{fig:phase_portraits}\hspace{-0.3mm}.IV, 
the conal order given by the cone field
$\delta \theta \geq 0$, $\delta \theta + \delta v \geq 0$ 
is not a partial order since, for any pair of points $x,y\in \mathbb{S}\times\real$, 
there exists a conal curve connecting $x$ to $y$ and viceversa. 
However, in a sufficiently small neighborhood of any point $x$, the conal order is a partial order.
\end{example}

\section{Differentially positive systems}
\label{sec:differentially_positive_systems}

\subsection{Definitions}
A dynamical system is differentially positive when 
its linearization is positive. Positivity is intended here in the sense of cone invariance \cite{Bushell1973}.
More precisely, a dynamical system $\Sigma$ on the conal state-input manifold 
$\calX\times\calU$ is differentially positive 
when the cone field
\begin{equation}
\label{eq:conal_field_XU}
\calK(x,u) = 
\underbrace{\calK_\calX(x,u)}_{\subseteq T_{x}\calX}
\times 
\underbrace{ \calK_\calU(x,u)}_{\subseteq T_{u}\calU} 
\subseteq T_{(x,u)}\calX \times \calU
\end{equation}
is invariant along the trajectories of the linearized system.
For discrete-time system $x^+ = f(x,u)$,
the invariance property has a simple formulation. 
The mapping
$f:\calX\times\calU \to \calX$, 
is differentially positive if, for all $x\in \calX$ and all $u,u^+ \in \calU$,
\begin{equation}
\label{eq:fundamental_discrete_condition}
 \partial f(x,u)  \calK(x,u)  \subseteq \calK_\calX(f(x,u),u^+) \ . 
\end{equation}
Indeed,
$\partial f(x,u)$ is a positive linear operator, mapping 
each tangent vector $(\delta x, \delta u) \in \calK(x,u) \subseteq T_{(x,u)}\calX\times\calU$
into $ \delta x^+ := \partial f(x,u)[\delta x,\delta u]  \in \calK_\calX(x^+,u^+) \subseteq T_{x^+}\calX$. 
A graphical representation for closed discrete systems is provided in 
Figure \ref{fig:discrete_differential_positivity}.
The relation between the positivity of the operator $\partial f(x,u)$ in \eqref{eq:fundamental_discrete_condition}
and the positivity of the linearization of $\Sigma$ is justified by the fact that
$\delta x^+ = \partial_x f(x,u) \delta x + \partial_u f(x,u) \delta u$,
which establishes the positivity of the linearized dynamics in the sense of
\cite{Bushell1973,Farina2000,DeLeenheer2001,Aliluiko2006}. 
\begin{figure}[htbp]
\centering
\includegraphics[width=0.63\columnwidth]{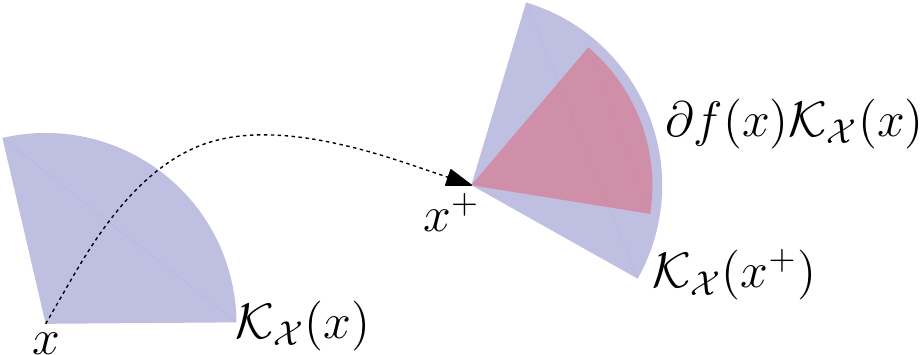}  
\caption{A graphical representation of condition 
\eqref{eq:fundamental_discrete_condition} 
for the (closed) discrete time system $x^+ = f(x)$.
{The cone field $\calK(x,u)$ reduces to  $\calK_\calX(x)$ in this case.}
}
\label{fig:discrete_differential_positivity}
\end{figure}

For general continuous-time
$\dot{x} = f(x,u) \in T_x\calX$ and discrete-time
$x^+ = f(x,u) \in \calX$ dynamical systems $\Sigma$ ($(x,u) \in \calX\times\calU$),
the definition of differential positivity involves
the \emph{prolonged system} $\delta \Sigma$ 
introduced in \cite{Crouch1987} 
\begin{equation}
\label{eq:vods}
\delta \Sigma :
\left\{
\begin{array}{rrcl}
(x^+) & \dot{x} &=& f(x,u) \\
(\delta x^+) & \dot{\delta x} &=& \partial_x f(x,u)\delta x + \partial_u f(x,u) \delta u \ . \\
\end{array}
\right.
\end{equation}
We call \emph{variational component} the second equation of \eqref{eq:vods}.

\begin{definition}
\label{def:diff+}
$\Sigma$ is a \emph{differentially positive} dynamical system 
(with respect to $\calK$ in \eqref{eq:conal_field_XU}) if, for all $t_0 \in \real$, any solution pair 
$((x,\delta x)(\cdot) , (u,\delta u)(\cdot) ): [t_0,\infty)  \to T\calX \times T\calU \in \delta\Sigma$
leaves the cone field $\cal{K}$ invariant. Namely, 
\begin{equation}
\label{eq:diff+}
\begin{array}{l}
\left\{
\begin{array}{rcl}
\delta x(t_0) &\in& \calK_\calX(x(t_0),u(t_0))  \\
 \ \delta u(t) &\in& \calK_\calU(x(t),u(t)) , \ \forall t\geq t_0
 \end{array}
 \right. \vspace{0.5mm}\\
\hspace{17mm} \Rightarrow \vspace{0.5mm} \\
\hspace{7.5mm}\delta x(t) \ \in \  \calK_\calX(x(t),u(t)) ,\ \forall t\geq t_0 \ . \vspace{-5mm}
\end{array}
\end{equation}
\end{definition} \vspace{1mm}

In continuous-time,
differential positivity of $\Sigma$ is thus positivity of the linearized system
$\dot{\delta x} = A(t) \delta x + B(t) \delta u$ along any solution pairs
$(x(\cdot),u(\cdot)) \in \Sigma$, where 
$A(t) :=  \partial_x f(x(t),u(t))$ and $B(t) := \partial_u f(x(t),u(t))$.
For closed systems $\dot{x} = f(x)$, with cone field $\calK_\calX(x)\subseteq T_x\calX$,
we have $\dot{\delta x} = A(t) \delta x$, where  $A(t) :=  \partial f(x(t))$,
along any given solution $x(\cdot):[t_0,\infty)\to\calX \in\Sigma$.
Therefore, the fundamental solution $\Psi_{x(\cdot)}(t,t_0)$ 
of the linearized dynamics \cite[Appendix C.4]{Sontag1998}
satisfies $\Psi_{x(\cdot)}(t,t_0)\calK_\calX(x) \subseteq \calK_\calX(x) $ 
for each $t\in[t_0,\infty)$, that is, $\Psi_{x(\cdot)}(t,t_0)$ is a positive linear operator. 

Strict differential positivity is to differential positivity what strict positivity
is to positivity. We anticipate that this
(mild) property will have a strong impact on the { asymptotic}
behavior of differentially positive systems, as shown
in Section \ref{sec:Perron-Forbenius}.
\begin{definition}
$\Sigma$ is (uniformly) \emph{strictly differentially positive} (with respect to $\calK$)
if differential positivity holds and 
there exists $T> 0$ and a cone field 
$\calR_\calX(x,u)\subseteq \mathrm{int}\calK_\calX(x,u)\cup\{0\}$ such that, for all $t_0 \in \real$,
any $((x,\delta x)(\cdot) , (u,\delta u)(\cdot) ): [t_0,\infty)  \to T\calX \times T\calU \in \delta\Sigma$
satisfies
\begin{equation}
\label{eq:USdiff+}
\begin{array}{l}
\left\{
\begin{array}{rcl}
\delta x(t_0) &\in& \calK_\calX(x(t_0),u(t_0))  \\
 \ \delta u(t) &\in& \calK_\calU(x(t),u(t)) , \ \forall t\geq t_0
 \end{array}
 \right. \vspace{0.5mm}\\
\hspace{17mm} \Rightarrow \vspace{0.5mm} \\
\hspace{8.5mm}\delta x(t) \in  \calR_\calX(x(t),u(t)) ,  \ {\forall t\geq t_0+T} \ . 
\end{array}
\end{equation}
We assume that the cone field $\calR_\calX$ also satisfies the following additional technical condition:
\begin{equation}
\label{eq:restriction_R}
\Gamma(x_1,u_1,x_2,u_2)\calR_\calX(x_1,u_1) = \calR_\calX(x_2,u_2) 
\end{equation}
for each  $(x_1,u_1),(x_2,u_2)\in\cal\calX\times\calU$, where
$\Gamma(x_1,u_1,x_2,u_2):T_{x_1}\calX \to T_{x_2}\calX$
is a linear invertible mapping such that 
$\Gamma(x_1,u_1,x_2,u_2)\calK_\calX(x_1,u_1) = \calK_\calX(x_2,u_2)$
(see Section \ref{sec:conal_manifold}).
\end{definition}

For open systems with output $h:\calX\times\calU \to \calY$ 
-- $\calY$ output manifold, endowed with the cone field $\calK_\calY(y)\in T_y\calY$
-- the notion of (strict) differential positivity requires the further condition that $h$
is a \emph{differentially positive mapping}, that is, 
$\partial h(x,u) \calK(x,u) \subseteq \calK_\calY(h(x,u))$, for each $(x,u)\in\calX\times\calU$.

\begin{remark}
\label{rem:geometric_conditions_for_positivity}
Differential positivity has a geometric characterization.  
restricting to closed systems for simplicity, consider the
cone field $ \calK_\calX(x) $ represented by \eqref{eq:basic_cone_field}
where $I$ is an index set and $k_i$ are smooth functions. 
Then, \eqref{eq:diff+} is equivalent to
require that $k_i(x(t),\delta x(t)) \geq 0$ along any solution 
$(x(\cdot),\delta x(\cdot)) \in \delta \Sigma$, for all $i\in I$.
Therefore, differential positivity for a discrete system can be established by 
testing that 
$\forall i \in I,\,k_i(x,\delta x) \geq 0$ implies 
$\forall i \in I,\,k_i^+ := k_i(f(x),\partial f(x)\delta x) \geq 0$,
for each $(x,\delta x) \in T\calX$.
In a similar way, for continuous systems, consider any pair 
$(x,\delta x)\in T\calX$ such that $k_i(x,\delta x) \geq 0$ for all $i\in I$
{and test that}, for any $j\in I$,
if $k_j(x,\delta x) = 0$ then $\dot {k}_j := \partial k_j(x,\delta x)[f(x),\partial f(x)\delta x] \geq 0$.
\end{remark}

\subsection{Examples}

\subsubsection{Positive linear systems are differentially positive}
\label{example:linear_positive_systems}

Consider the dynamics $\Sigma$ given by 
 $x^+ = A x$ on the vector space $\calV$.
 Positivity with respect to the cone $\calK_\calV\subseteq\calV$
reads $A\calK_\calV \subseteq\calK_\calV$, \cite{Bushell1973}. 
A typical example is provided by the case of a matrix $A$ with non-negative entries
which guarantees the invariance of the positive orthant $\calK_\calV := \real^n_+$.

Since each tangent space of a vector space can be identified to the vector space itself,
i.e. $T_x \calV= \calV$ for each $x \in \calV$, 
consider the manifold $\calX := \calV$ and define the lifting of the cone $\calK_\calV$ to the
cone field $\calK_{\calX}(x) := \calK_\calV\subseteq T_x\calX$, for each $x \in \calX$ (\emph{constant} cone field).
Then the linearized dynamics reads $\delta x^+ = A \delta x$ and
the prolonged system trivially satisfies $A \calK_\calX(x) \subseteq \calK_\calX(Ax)$.

\subsubsection{Monotone systems are differentially positive}
\label{example:monotone_systems}
A monotone dynamical system \cite{Smith1995,Angeli2003}
is a dynamical system whose trajectories preserve some 
partial order relation on the state space. 
Moving from closed \cite{Smith1995,Hirsch1995,Hirsch1988,Dancer1998,Hirsch2003,Piccardi2002} 
to open systems \cite{Angeli2003,Angeli2004a,Angeli2004b}, 
this wide class of systems is extensively adopted in biology and chemistry
both for modeling and control
\cite{DeLeenheer2004,Enciso2005, DeLeenheer2007,Sontag2007,Angeli2008,Angeli2012}.

The partial order $\preceq$ of a monotone system
is typically induced by a conic subset $\calK_\calV\subseteq\calV$ 
of the state (vector) space $\calV$. Precisely,
two points $x,\hat{x}\in \calV$ satisfy $x\preceq_{\calK_\calV} \hat{x}$
if and only if 
$\hat{x}\!-\!x \in \calK_\calV$. 
The preservation of the order along the system dynamics 
reads as follows: if $x(\cdot),\hat{x}(\cdot)\in \Sigma$
satisfy $x(t_0) \preceq_{\calK_\calV} \hat{x}(t_0)$ for some initial time $t_0$,
then $x(t) \preceq_{\calK_\calV} \hat{x}(t)$ for all $t\geq t_0$, \cite{Smith1995}.

To show that a monotone system is differentially positive, 
consider 
{$\calV$ as a manifold
endowed with the 
constant cone field $\calK_\calV(x) := \calK_\calV$, $x \in \calV$.
By monotonicity, the infinitesimal difference between two ordered {neighboring} solutions
$\delta x(\cdot) := \hat{x}(\cdot)-x(\cdot)$ satisfies 
$\delta x(t) \in {\calK_\calV}(x(t))$, for each $t\geq t_0$. Differential positivity follows from the fact
that $(x(\cdot),\delta x(\cdot))$ is a trajectory of the prolonged system $\delta \Sigma$.

\begin{theorem}
Given any cone $\calK_\calV$ on the vector space $\calV$, the partial order
$\preceq_{\calK_\calV}$, and the cone field $\calK_\calX(x) := \calK_\calV$,
a (closed) dynamical system is monotone if and only if is differentially positive.
\end{theorem}
\begin{proof}
For constant cone fields on vector spaces recall that 
$\preceq_{\calK_\calV}$ and $\sqsubseteq_{\calK_\calV}$ are equivalent
relations (see Section \ref{sec:conal_manifold}). 
Consider a conal curve $\gamma(t_0,\cdot)$ connecting two ordered initial points
$\gamma(t_0,0) := x(t_0) \preceq_{\calK_\calV} \hat{x}(t_0) =: \gamma(t_0,1)$. 
Note that $\gamma(t_0,s_1) \sqsubseteq_{\calK_\calV} \gamma(t_0,s_2)$ 
for each $s_1\leq s_2$. For each $s\in[0,1]$, 
let $\gamma(\cdot,s)$ be a trajectory of $\Sigma$. Indeed,
$\gamma(t,\cdot)$ represents the time evolution of the curve $\gamma(t_0,\cdot)$
along the flow of the system.
[$\Leftarrow$] 
Differential positivity guarantees that $\gamma(t,\cdot)$ is a conal curve for each $t\geq t_0$.
This follows from the fact that 
the pair $(x_s(t), \delta x_s(t)) := (\gamma(t,s), \frac{d}{ds} \gamma(t,s))$ 
is a trajectory of the prolonged system 
$\delta \Sigma$ for each $s\in [0,1]$. 
Thus, $x(t) \sqsubseteq_{\calK_\calV} \hat{x}(t)$ for all $t\geq t_0$.
[$\Rightarrow$]
Monotonicity guarantees that $\gamma(t,s_1)\preceq_{\calK_\calV} \gamma(t,s_2)$
for all $s_1 \leq s_2$. By a limit argument, $\frac{d}{ds} \gamma(t,s) \in \calK_\calV(\gamma(t,s))$
for all $t\geq t_0$ and al $s \in [0,1]$. Thus $\gamma(t,\cdot)$ is a conal curve. 
Note that $(x_s(t), \delta x_s(t)) := (\gamma(t,s), \frac{d}{ds} \gamma(t,s))$ 
is a trajectory of the prolonged system. Since $\gamma(t_0,\cdot)$ is a generic
conal curve, \eqref{eq:diff+} follows.
\end{proof}
}

A similar result holds for open monotone systems,
which are typically characterized by introducing 
two orders
$\preceq_{\calK_\calX}$ and $\preceq_{\calK_\calU}$, respectively induced by the cone
$\calK_\calX$ on the state space $\calX$ and $\calK_\calU$ on the input space $\calU$,
\cite[Definition II.1]{Angeli2003}. Extending the argument above it is possible to
show that a dynamical system $\Sigma$ is
monotone with respect to $(\preceq_{\calK_\calX},\preceq_{\calK_\calU})$
\emph{if and only if} $\Sigma$ is differentially positive on the \emph{vector space}
$\calX\times\calU$ endowed with the \emph{constant} cone field
$\calK(x,u) := \calK_\calX \times\calK_\calU$, for each $(x,u) \in \calX\times\calU$.
In this sense, differential positivity on vector spaces and constant cone fields
is the differential formulation of monotonicity. 

\subsubsection{Differential positivity of cooperative systems and the Kamke condition}
A cooperative system $\dot{x} = f(x)$ with state space $\calX := \realn$
is monotone with respect to the partial order induced by the positive orthant 
$\real^n_+$, thus differentially positive with respect to the cone field
$\calK_\calX(x) := \real^n_+$, $x\in \calX$.
Exploiting the geometric conditions of Remark \ref{rem:geometric_conditions_for_positivity},
differentially positivity with respect to $\calK_\calX$ holds when 
\begin{equation}
\label{eq:cooperative_firstexample}
 [\partial f(x)]_{ij}  \geq  0 \qquad \mbox{for all } 1\leq i\neq j\leq n ,\, x\in \calX,
\end{equation}
where $[\cdot]_{ij}$ denotes the $ij$ component. 
To see this, define $E_i$ as the vector whose $i$-th element is equal to one and the remaining 
to zero and note that the positive orthant is defined by the set of $\delta x$ that satisfy
$\langle E_i, \delta x  \rangle \geq 0$.
Then, from Remark \ref{rem:geometric_conditions_for_positivity}, 
the invariance reads 
$\langle E_i, \delta x  \rangle = 0 \Rightarrow \langle E_i, \partial f(x) \delta x   \rangle  \geq 0 $
for any $x \in \calX$, $\delta x \in \calK_\calX(x) = \real^n_+$, and 
$i \in \{1,\dots, n\}$. \eqref{eq:cooperative_firstexample} 
follows by selecting $\delta x = E_j \neq E_i$.
Indeed, $\partial_x f(x)$ is a Metzler matrix for each $x\in\calX$
\cite[Section VIII]{Angeli2003}.

Cooperative systems typically satisfy \eqref{eq:cooperative_firstexample}, 
as shown in \cite[Remark 1.1]{Smith1995} on closed systems.
A similar result is provided in \cite[Proposition III.2]{Angeli2003} for open systems. 
In this sense, the pointwise geometric conditions in Remark
\ref{rem:geometric_conditions_for_positivity} revisit and extend the comparison 
between cooperative systems, 
incrementally positive systems of \cite[Section VIII]{Angeli2003}, and the 
Kamke condition provided in \cite[Chapter 3]{Smith1995}. 

\subsubsection{One dimensional continuous-time systems are differentially positive}
\label{sec:one_dimensional}
This property is well-known for systems in $\real$: solutions are partially ordered because 
they cannot ``pass each other''. It remains true on closed manifolds such as $\mathbb{S}$,
even though the conal order does not induce a (globally defined) partial order in that case. 

\subsubsection{Non-constant cones for oscillating dynamics}
\label{sec:first_nonlinear_oscillator}
Moving from constant to non-constant cone fields opens 
the way to the analysis of more general limit sets such a oscillations 
or limit cycles.
The harmonic oscillator studied in Section \ref{sec:nutshell}
provides a first simple example of differential positivity with respect to
a non-constant cone field. In particular, consider
{
$\calK_\calX(x) := 
\{ 
\delta x \in \real^2\setminus\{0\}
\,|\,
k_1(x,\delta x)  \geq 0 ,\,
 k_2(x,\delta x) \geq 0 
\} 
$,
where
$k_1(x,\delta x) := -(x_1+x_2) \delta x_1 + (x_1-x_2)\delta x_2$ and
$k_2(x,\delta x) := - (x_2-x_1)\delta x_1 + (x_1+x_2) \delta x_2$.
The cone field is well defined on the (invariant) manifold $\calX := \real^2 \setminus \{0\}$.}
Differential positivity with respect to $\calK_\calX(x)$ follows from the geometric conditions in 
Remark \ref{rem:geometric_conditions_for_positivity},
since $\dot{k}_1 = 0$ and $\dot{k}_2  = 0$ everywhere. 

The differential positivity of the harmonic oscillator with respect to $\calK_\calX(x)$
is not surprising if one looks at the representation of the oscillator in polar coordinates
$\dot{\vartheta} = 1$, $\dot{\rho} = 0$.
The state manifold becomes the cylinder $\mathbb{S} \times \real_{+}$
and the system decompose into two one-dimensional systems, 
which suggests the invariance of any cone field rotating with $\vartheta$, 
as shown in Figure \ref{fig:ho_invariance} (left). Indeed, 
polar coordinates suggest differential positivity for arbitrary 
decoupled dynamics $\dot{\vartheta} = f(\vartheta)$, $\dot{\rho} = g(\rho)$
with respect to the cone field
$
 \calK(\vartheta,\rho) := \{(\delta \vartheta, \delta \rho) \in \real^2)\,|\, 
 \delta \vartheta \geq 0 \,,\
 \delta \rho \geq 0\}
$.
In fact, the linearization reads $\dot{\delta \vartheta} = \partial f(\vartheta)  \delta \vartheta$,
$\dot{\delta \rho} = \partial g(\rho)  \delta \rho$, which guarantees
that $\dot{\delta \vartheta} = 0$ for $\delta \vartheta = 0$ and 
$\dot{\delta \rho} = 0$ for $\delta \rho= 0$,
as required by Remark \ref{rem:geometric_conditions_for_positivity}.

Possibly, the invariance of the cone field can be strengthened to 
contraction by combining the two uncoupled dynamics.
For example, when $f(\vartheta) = 1$ and $g(\rho) = \rho-\frac{\rho^3}{3}$,
the trajectories of the variational dynamics move towards the interior of the cone field
$
 \calK(\vartheta,\rho) := \{(\delta \vartheta, \delta \rho) \in \real^2)\,|\, 
 \delta \vartheta \geq 0 \,,\
 \delta \vartheta^2 - \frac{\delta \rho^2}{\rho^2} \geq 0\}
$.
In fact, $\frac{d}{dt} (\delta \vartheta^2 - \frac{\delta \rho^2}{\rho^2}) = (1+\frac{\rho^2}{3}) \frac{\delta \rho^2}{\rho^2} > 0$
for each $(\delta \vartheta,\delta \rho)\in \mathrm{bd}\calK(\vartheta,\rho)\setminus\{0\}$.
Figure \ref{fig:ho_invariance} (right) provides a representation of the (projective) contraction
of the cone. We anticipate that this contraction property is tightly connected to the existence 
of a globally attractive limit cycle.

\begin{figure}[htbp]
\centering
\includegraphics[width=0.45\columnwidth]{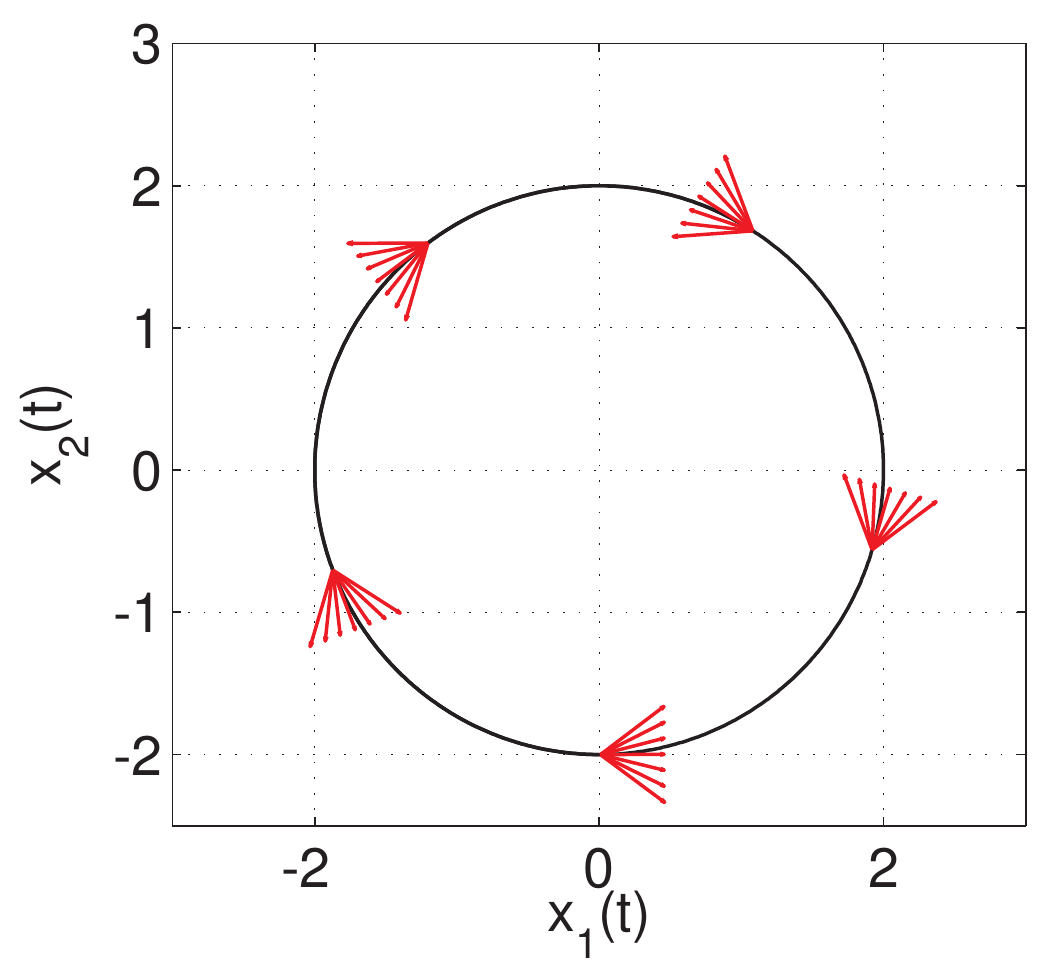}  
\includegraphics[width=0.45\columnwidth]{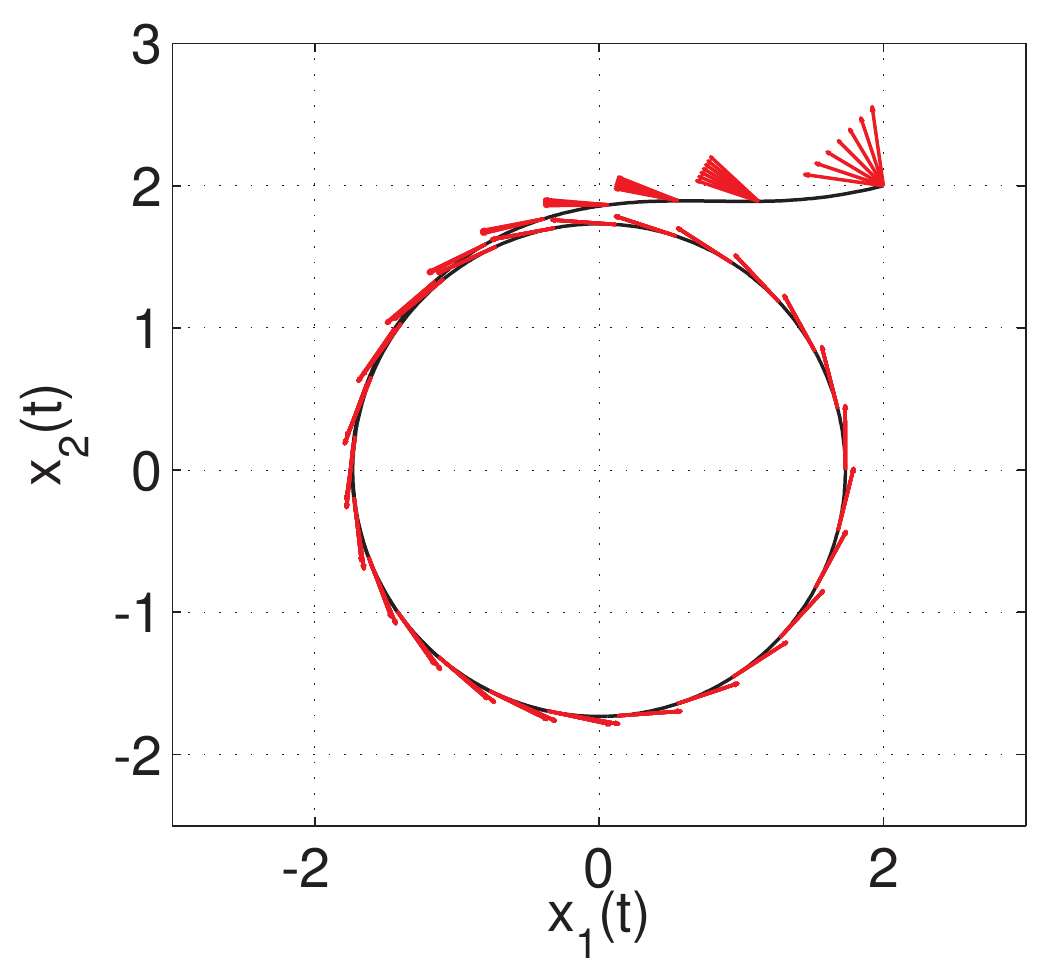}  
\caption{Simulations of the variational dynamics $\delta \Sigma$ 
of the two cases (i)~$\dot{\rho} = 0$ (left) and (ii)~$\dot{\rho} = (\rho-\frac{\rho^3}{3})$ (right),
from the same initial condition.
The red arrows represent the motion of the variational component
along the trajectory. 
}
\label{fig:ho_invariance}
\end{figure}
{
\begin{remark}
Differential positivity requires 
classical positivity of the linearized dynamics at fixed points. 
In fact,  
Definition \ref{def:diff+} shows that the cone field at any fixed point
$x^*$ is given by the invariant cone of the (positive) linearized dynamics at $x^*$.
The harmonic oscillator is not a positive linear system, because of 
the presence of the two complex eigenvalues. Thus, it is not
differentialy positive in $\real^2$.
However, polar coordinates reveal that it is
differentially positive in the manifold
$\calX = \real^2 \setminus\{0\}$.
\end{remark}
}

\section{Differential Perron-Frobenius theory}
\label{sec:Perron-Forbenius}

\subsection{Contraction of the Hilbert metric}
\label{subsec:lifted_Hilbert}
\label{subsec:variational_Perron-Forbenius}

Bushell \cite{Bushell1973} (after Birkhoff \cite{Birkhoff1957})
used the Hilbert metric on cones to show that the 
strict positivity of a mapping guarantees contraction
among the rays of the cone, 
opening the way to many  contraction-based results
in the literature of positive operators \cite{Bushell1973,Nussbaum1994,Sepulchre2010,Bonnabel2011,Lemmens2012},
among which the reduction of the Perron-Frobenius theorem to a special
case of the contraction mapping theorem \cite{Pratt1982,Bushell1973,Birkhoff1957}.
Taking inspiration from these important results, 
we rely on the infinitesimal contraction 
properties of the Hilbert metric 
to study the contraction properties of differentially positive systems.

Consider the product manifold $\calX \times \calU$ 
where $\calX$ is endowed with the cone field $\calK_\calX(x,u) \subseteq T_x\calX$,
for each $(x,u) \in \calX\times \calU$.
Following \cite{Bushell1973},
for any given $(x,u)$,
take any $\delta x, \delta y \in \calK_\calX(x,u)\setminus\{0\}$ and define 
the quantities 
\begin{equation}
\label{eq:Mm}
\begin{array}{c}
\hspace{-2.5mm} M_{\calK_\calX(x,u)}(\delta x,\delta y) :=
 \inf \{\lambda \in \real_{\geq 0} \,|\, \lambda \delta y - \delta x \in \calK_\calX(x,u) \} \vspace{1mm}\\
 m_{\calK_\calX(x,u)}(\delta x, \delta y) := 
 \sup \{\lambda \in \real_{\geq 0} \,|\, \delta x - \lambda \delta y  \in \calK_\calX(x,u) \} .
 \end{array}
\end{equation}
$M_{\calK_\calX(x,u)}(\delta x,\delta y) := \infty$
when $ \{\lambda \in \real_{\geq 0} \,|\, \lambda \delta y - \delta x \in \calK_\calX(x,u) \} = \emptyset$.
The Hilbert's metric $d_{\calK_\calX(x,u)}$ induced by $\calK_\calX(x,u)$ is given by
\begin{equation}
\label{eq:lifted_hilbert_metric}
 d_{\calK_\calX(x,u)}( \delta x, \delta y) = \log \left(\frac{ M_{\calK_\calX(x,u)}(\delta x,\delta y)}{ m_{\calK_\calX(x,u)}(\delta x,\delta y)}\right) \ .
\end{equation}
{
In each cone $\calK_\calX (x,u)$
$d_{\calK_\calX(x,u)}$ is a projective distance:
for each $\delta x,\delta y,\delta z \in \calK_\calX (x,u)$,
$d_{\calK_\calX(x,u)}(\delta x,\delta y) \geq 0$, 
$d_{\calK_\calX(x,u)}(\delta x,\delta y) = d_{\calK_\calX(x,u)}(\delta y,\delta x)$, 
$d_{\calK_\calX(x,u)}(\delta x,\delta y) \leq d_{\calK_\calX(x,u)}(\delta x,\delta z)+d_{\calK_\calX(x,u)}(\delta z,\delta y)$,
and 
$d_{\calK_\calX(x,u)}(\delta x,\delta y) = 0$
if and only if $\delta x=\lambda \delta y$ with $\lambda \geq 0$.
}

The following theorem is a generalization of Birkhoff result: it shows that
strict differential positivity guarantees the \emph{exponential contraction} 
of the metric when the input $u(\cdot)$ acts uniformly on the 
system (a feedforward signal).
The uniform action of the input is modeled by taking the variational input $\delta u(\cdot) = 0$,
since $\delta u$ represents the infinitesimal mismatch between two inputs.

For readability, in what follows we denote the Hilbert metric along a solution pair
$d_{\calK_\calX(x(t),u(t))}(\cdot,\cdot)$ with $d_{*(t)}(\cdot,\cdot)$.
{
\begin{theorem}
\label{thm:dsch}
Let $\Sigma$ be a dynamical system on the state/input manifold $\calX\times\calU$,
 differentially positive with respect to the cone field
$\calK(x,u) := \calK_\calX(x,u)\times \{0\}$,
where  $\calK_\calX(x,u) \subseteq T_x\calX$ for each $(x,u)\in \calX\times\calU$.
Then, for all $t\geq t_0$,
\begin{equation}
\label{eq:Hilbert_invariance}
 d_{*(t)}(\delta x_1(t),\delta x_2(t)) \leq d_{*(t_0)} (\delta x_1(t_0),\delta x_2(t_0)) 
\end{equation}
for any $(x(\cdot),\delta x_1(\cdot) , u(\cdot),0), (x(\cdot),\delta x_2(\cdot) , u(\cdot),0) \in \delta\Sigma$ 
with domain $[t_0,\infty)$ and
$\delta x_1(t_0),\delta x_2(t_0) \in \calK_\calX(x(t_0),u(t_0))$.

If $\Sigma$ is strictly differentially positive then
there exist $\rho\geq 1$ and $\lambda > 0$ such that, 
for all $t\geq t_0$, 
\begin{equation}
\label{eq:Hilbert_contraction}
 d_{*(t)}(\delta x_1(t),\delta x_2(t)) \leq \rho e^{-\lambda(t-t_0)} d_{*(t_0)} (\delta x_1(t_0),\delta x_2(t_0)) 
\end{equation}
for any $(x(\cdot),\delta x_1(\cdot) , u(\cdot),0), (x(\cdot),\delta x_2(\cdot) , u(\cdot),0) \in \delta\Sigma$ 
with domain $[t_0,\infty)$ and
$\delta x_1(t_0),\delta x_2(t_0) \in \calK_\calX(x(t_0),u(t_0))$.
Moreover, $d_{*(t)}(\delta x_1(t),\delta x_2(t)) < \infty$ for $t \geq t_0+T$.
\end{theorem}
}

\subsection{The Perron-Frobenius vector field}
\label{sec:PF_vectorfield}

The Perron-Frobenius vector of a strictly positive linear map is a fixed point of the projective space. 
Its existence is a consequence of the 
contraction of the Hilbert metric, \cite{Bushell1973}. 
To exploit the generalized contraction of Theorem \ref{thm:dsch}, 
we assume that the input acts uniformly on the system, that is, $\delta u(\cdot) = 0$.
We endow the state manifold $\calX$ with a (smooth) Riemannian structure
{ and we define 
$\calB(x) := \{\delta x \in T_x\calX \,|\, |\delta x|_x=1\}\subseteq T_x\calX$,
to make the following assumption.
\begin{assumption}
\label{assume:completeness}
$(\calK_\calX(x,u)\cap \calB(x),d_{\calK_\calX(x,u)})$ is a \emph{complete metric space}
for any given $(x,u)\in \calX\times\calU$  \footnote{
The reader is referred to \cite[Section 4]{Bushell1973}, \cite[Section 2.5]{Lemmens2012},
or \cite{Zhai2011} for examples of complete metric spaces on cones.
}. 
\end{assumption}

To introduce the \emph{Perron-Frobenius vector field} we 
study the { asymptotic behavior} of $\delta \Sigma$,}
looking at solutions pairs
$(z(\cdot),u(\cdot))\in \Sigma$ with domain $I := (-\infty,t)$ (backward completeness of $\Sigma$).
Recall that for any $(z(\cdot),u(\cdot)) \in \Sigma$, if $\delta u(\cdot) = 0$ then
$\delta z(t) = \partial_{z(t_0)}\psi(t,t_0,z(t_0),u(\cdot))\delta z(t_0)$ is a trajectory
of the variational components of $\delta \Sigma$ along $(z(\cdot),u(\cdot))$.

\begin{theorem}
\label{thm:differential_pf}
Let $\Sigma$ be a dynamical system on the state/input manifold $\calX\times\calU$.
{Suppose that} $\Sigma$ is strictly differentially positive with respect to the cone field 
 $\calK(x,u) := \calK_\calX(x,u)\times \{0\} $
such that $\calK_\calX(x,u) \subseteq T_x\calX$ for each $(x,u)\in \calX\times\calU$
and suppose that  Assumption \ref{assume:completeness} holds.

For any input $u(\cdot):\real\to\calU$ which makes $\Sigma$ backward complete, 
there exists a time-varying vector field $\mathbf{w}_{u(\cdot)}(x,t) \in \mathrm{int} \calK_\calX(x,u(t)) \cap \calB(x)$,
$x\in \calX$ and $t\in\real$, such that any solution pair $(z(\cdot),u(\cdot)) \in \Sigma$ satisfies
\begin{equation}
\label{eq:PF_vector_field}
\begin{array}{c}
\lim\nolimits\limits_{t_0\to -\infty} \partial_{z(t_0)}\psi(t,t_0,z(t_0),u(\cdot))\calK_{\calX}(z(t_0),u(t_0)) =\\
\hspace{34mm} = \{ \lambda \mathbf{w}_{u(\cdot)}(z(t),t) \,|\, \lambda \geq 0 \} \ .
\end{array}
\end{equation}
We call this vector field the \emph{Perron-Frobenious vector field.} 
\end{theorem}
\begin{corollary}
\label{cor:differential_time-invariant_pf}
Under the assumptions of Theorem \ref{thm:differential_pf},
if $u(\cdot) = u\in\calU$ (constant), then the Perron-Frobenius vector field
reduces to a continuous time-invariant vector field $\mathbf{w}_u(x)$.
For linear systems the { Perron-Frobenius} vector field reduces to the (constant)
Perron-Frobenius vector.
\end{corollary}

The evolution of an initial cone $\calK_\calX(z(t_0),u(t_0))$ along the (variational) flow of the system 
asymptotically converges to the span of the Perron-Frobenius vector field $\mathbf{w}_{u(\cdot)}(x,t)$ attached to each $x\in \calX$,
as illustrated in Figure \ref{fig:thm3_illustration}. 
Decomposing the variational trajectory $\delta z(\cdot)$ along $z(\cdot)$ into
a \emph{directional} component 
$\vartheta(t) := \frac{\delta z(t)}{|\delta z(t)|_{z(t)}} \in \calK_\calX(z(t),u(t))\cap \calB_{z(t)}$, 
and a \emph{magnitude} component $\rho(t):=|\delta z(t)|_{z(t)}$,
Theorem \ref{thm:differential_pf} establishes that 
$\vartheta(t)$ is { guaranteed} to converge
to $\mathbf{w}_{u(\cdot)}(z(t),t)$, for any initial condition $\vartheta(t_0)$.  

For constant inputs the Perron-Frobenius vector field has a simple geometric characterization.
Take any trajectory $(x(\cdot),\delta x(\cdot))$ of the prolonged system $\delta \Sigma$
under the action of the constant input $u$, and suppose that 
$d_{(x(t),u)} (\mathbf{w}_u(x(t)),\delta x(t)) = 0$ for some $t\in \real$. Then,
from \eqref{eq:Hilbert_contraction}, 
$d_{(x(t+\tau),u)} (\mathbf{w}_u(x(t+\tau)),\delta x(t+\tau)) = 0$ for each $\tau \geq 0$, 
which shows that $\mathbf{w}_u(x(t))$ must be a time-reparametrized trajectory of $\Sigma$.
Therefore, $\mathbf{w}_u(x)$ belongs to $\mathrm{int} \calK_\calX(x,u)$ for all $x$ and satisfies
the partial differential equation
$[\partial_x \mathbf{w}_u(x)]f(x,u) = [\partial_x f(x,u)] \mathbf{w}_u(x) - \lambda(x,u) \mathbf{w}_u(x)$
for continuous time systems, for some $\lambda(x,u)\in\real$ which guarantees
$|\mathbf{w}_u(x)|_x = 1$. In a similar way, for discrete dynamics we have
$\mathbf{w}_u(f(x,u)) = \lambda(x,u) [\partial_x f(x,u)]w(x,u)$. As before, 
$\lambda(x,u)\in\real$ is selected to guarantee $|\mathbf{w}_u(x)|_x = 1$.
Existence and uniqueness of the solution $\mathbf{w}_u(x)$ follow from the
contraction of the Hilbert metric,
under the assumption of backward and forward invariance of
$\calX$.

\begin{figure}[htbp]
\centering
\includegraphics[width=0.96\columnwidth]{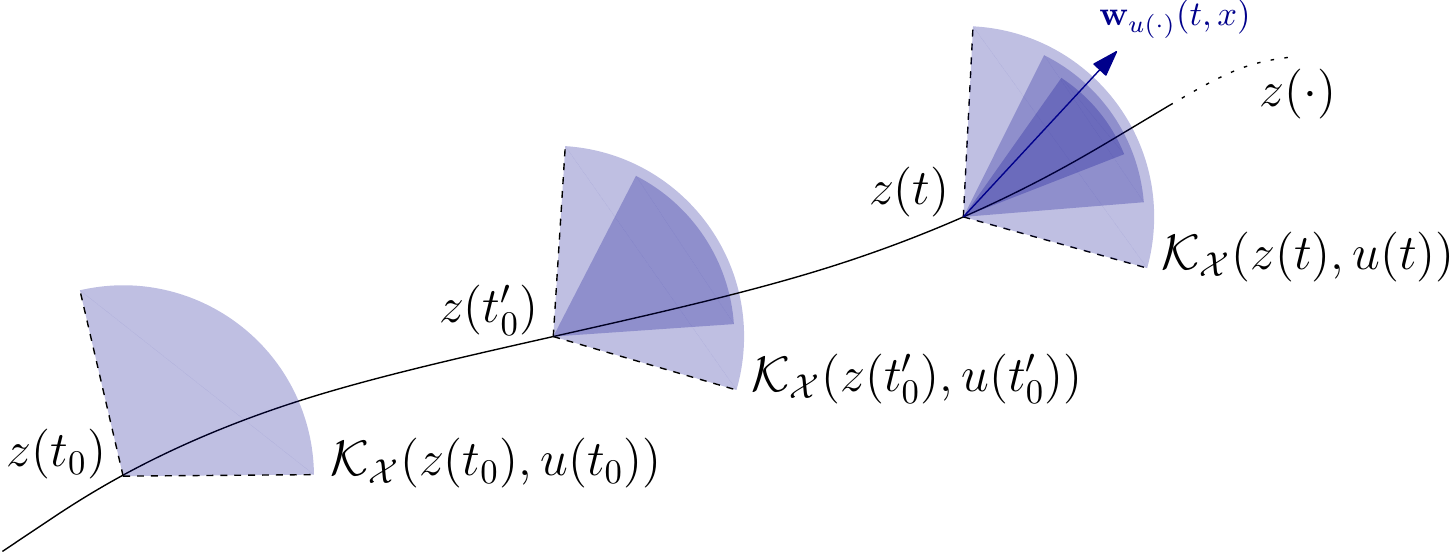}  
\caption{The contraction at time $t$ from different initial cones, for $t_0 < t_0' < t$.  
Note that 
$\partial_{z(t_0)}\psi(t,t_0,z(t_0),u(\cdot))\calK_{\calX}(z(t_0),u(t_0)) \subseteq \partial_{z(t_0')}\psi(t,t_0',z(t_0'),u(\cdot))\calK_{\calX}(z(t_0'),u(t_0')) \subseteq \calK_\calX(z(t),u(t))$.
At time $t$, for $t-t_0 \to \infty$, the cone reduces to a line.}
\label{fig:thm3_illustration}
\end{figure}

\section{Limit sets of (closed) \\ differentially positive systems}
\label{sec:limit_sets}

\subsection{Behavior dichotomy}
\label{sec:main_theorem}
For closed continuous-time
dynamical systems (or open continuous-time systems
with constant inputs)
the combination of the local order on the system 
state manifold and the projective contraction of the 
variational dynamics toward the Perron-Frobenius 
vector field $\mathbf{w}(x)$
{restrict the asymptotic} behavior
of differentially positive systems. 
The next theorem characterizes the $\omega$-limit
sets of those systems.

\begin{theorem}
\label{thm:limit_sets}
Let $\Sigma$ be a 
closed continuous (complete) system $\dot{x}=f(x)$
with state manifold $\calX$,
 strictly differentially positive 
with respect to the cone field $\calK_\calX(x)\subseteq T_x \calX$.
Under Assumption \ref{assume:completeness}, suppose that
the trajectories of $\Sigma$ are bounded. Then, 
for every $\xi\in \calX$,
the $\omega$-limit set $\omega(\xi)$ satisfies
one of the following two properties:
\begin{itemize}
\item[(i)] The vector field $f(x)$ is aligned with the Perron-Frobenius vector field $\mathbf{w}(x)$
for each $x\in \omega(\xi)$ {(i.e. $f(x)=\lambda(x) \mathbf{w}(x)$, $\lambda(x) \in \real$)}, and $\omega(\xi)$ is either a fixed point 
or a limit cycle or a set of fixed points and connecting arcs; 
\item[(ii)] The vector field $f(x)$ is { not} aligned with the 
Perron-Frobenius vector field $\mathbf{w}(x)$ for each $x \in \omega(\xi)$ 
{ such that $f(x)\neq 0$}, 
and either 
$
\liminf\nolimits\limits_{t\to\infty} 
|\partial_x \psi(t,0,x)\mathbf{w}(x)|_{\psi(t,0,x)} = \infty
$
or 
$\lim\nolimits\limits_{t\to\infty} f(\psi(t,0,x)) = 0$.
\vspace{-4mm}
\end{itemize}
\end{theorem}

The interpretation of Theorem \ref{thm:limit_sets} is that
the {asymptotic} behavior of $\Sigma$ is either 
described by a 
\emph{Perron-Frobenius curve} $\gamma^{\mathbf{w}}(\cdot)$,
that is, a curve $\dot{\gamma}^{\mathbf{w}}(s) = \mathbf{w}(\gamma^{\mathbf{w}}(s))$ for all $s\in \dom \gamma^{\mathbf{w}}(\cdot)$;
or is the union of the limit points 
of some trajectory $\psi(\cdot,0,\xi)$,
$\xi\in \calX$, 
nowhere tangent to the Perron-Frobenius vector field,
as clarified in Section \ref{sec:complex_attractors},
and characterized by high sensitivity with respect to 
initial conditions, because of the unbounded linearization.
The proof of Theorem \ref{thm:limit_sets} in Appendix, Section \ref{sec:proof_main_theorem},
is of interest on its own
since it illustrates how differential Perron-Frobenius theory
impacts the behavior of $\Sigma$.
In the next two subsections we further discuss the implications
of Theorem \ref{thm:limit_sets} in case (i) and in case (ii),
respectively.

\subsection{Simple attractors of differentially positive systems}
\label{sec:simple_attractors}
A first consequence of Theorem \ref{thm:limit_sets} is a result
akin to Poincare-Bendixson characterization of limit sets of planar systems.

\begin{corollary}
\label{thm:PBtheorem}
Under the assumptions of Theorem \ref{thm:limit_sets},
consider an open, forward invariant region $\calC \subseteq \calX$
that does not contain any fixed point. 
If the vector field $f(x) \in \mathrm{int}\calK_\calX(x)$ for any $x\in \calC$, then
there exists a unique attractive periodic orbit contained in $\calC$.
\end{corollary}

The result shows the potential of differential positivity
for the analysis of limit cycles in possibly high dimensional spaces. 
Since stable limit cycles must 
correspond to Perron-Frobenius curves, stable limit cycles are
excluded when Perron-Frobenius curves are open, a property
always satisfied in vector spaces with constant cone field.
For a differentially positive system defined in a vector space, 
the cone field must necessarily ``rotate'' with the periodic orbit 
in order to allow for limit cycle attractors (see, for example,
Section \ref{sec:first_nonlinear_oscillator}).

Beyond isolated fixed point and limit cycles, the limit sets of 
differentially positive systems are severely restricted by (local)
order properties, see Figure \ref{fig:local_order} for an illustration.
In particular, the intuitive argument ruling out homoclinic orbits like in
Figure \ref{fig:homoclinic} is made rigorous with
Theorem \ref{thm:limit_sets}. 
A limit set given by a connecting arc between two
hyperbolic fixed points can exists only 
if it is everywhere tangent 
to the Perron-Frobenius vector field
(Theorem \ref{thm:limit_sets}.i),
or nowhere tangent 
to the Perron-Frobenius vector field
(Theorem \ref{thm:limit_sets}.ii). 
Because any orbit between two
hyperbolic fixed points must belong to the unstable
manifold of its $\alpha$-limit set and to the stable manifold
of its $\omega$-limit set, it can be a Perron-Frobenius curve
only if, whenever it is tangent to the Perron-Frobenius eigenvector
of its $\alpha$-limit, it is also tangent to the 
Perron-Frobenius eigenvector of its $\omega$-limit.
\begin{corollary}
\label{thm:homoclinic}
Under the assumptions of Theorem \ref{thm:limit_sets},
consider an orbit that connects two hyperbolic fixed points $y_e$, $z_e$,
respectively as $t \to -\infty$ and $t \to \infty$.
If the orbit is tangent to $\mathbf{w}(y_e)$ at $y_e$,
then it is tangent to $\mathbf{w}(z_e)$ at $z_e$.
\end{corollary}

The corollary rules out the possibility of a homoclinic orbit 
with a one-dimensional unstable manifold, a typical ingredient of 
strange attractors.
For system depending on parameters, the corollary rules out 
the possibility of homoclinic bifurcations 
\cite[Chapter 8]{Strogatz1994} where the homoclinic orbit
is tangent to the dominant eigenvector of the saddle point.
In accordance with Theorem \ref{thm:limit_sets}, 
a limit set given by a homoclinic orbit can only exist 
if it is nowhere tangent 
to the Perron-Frobenius vector field, which rules out
the possibility of being part of a simple attractor.
The two situations are illustrated in Fig \ref{fig:homoclinic2}.

\begin{figure}[htbp]
\centering
\includegraphics[width=0.8\columnwidth]{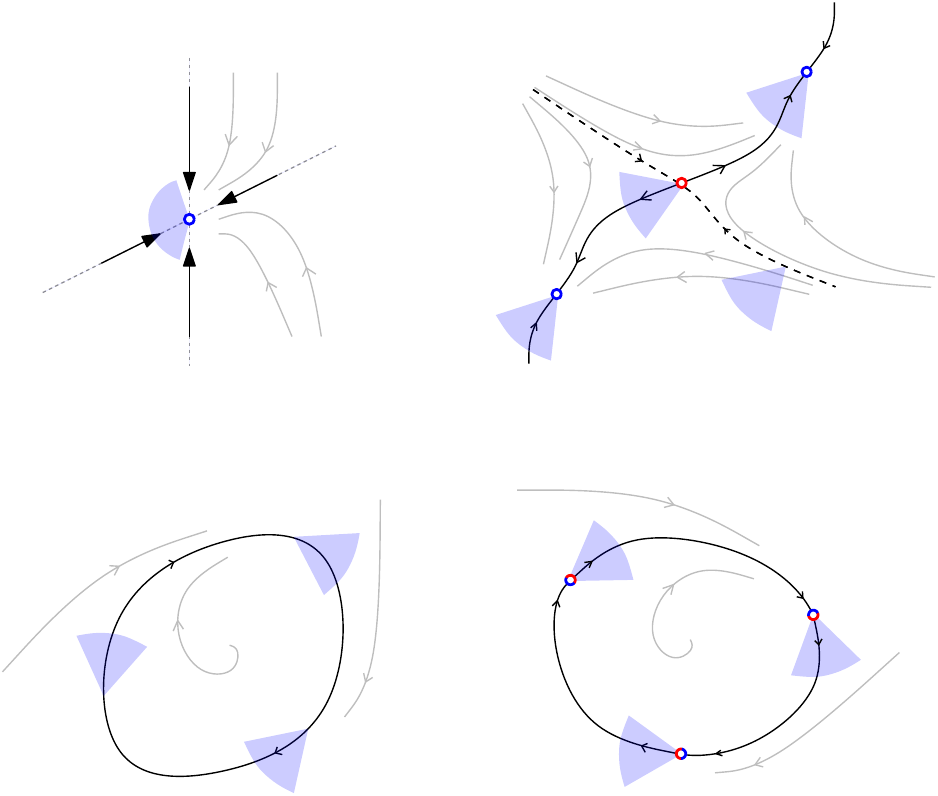} 
\caption{Simple attractors consistent with the ordering
property of a Perron-Frobenius curve.
Cones are shaded in blue. 
Top-left: fixed point with dominant eigenvector. 
Top-right: bistable systems, the heteroclinic orbits 
connecting the saddle (red) to the stable fixed points (blue) 
coincides with the image of some Perron-Frobenius curves.
Bottom: limit cycles (left) 
or fixed points and connecting arcs (right)
coincides with the image of some closed Perron-Frobenius curve,
thus require require a rotating cone field (on vector spaces).
}
\label{fig:local_order}
\end{figure}

\begin{figure}[htbp]
\centering
\includegraphics[width=0.46\columnwidth]{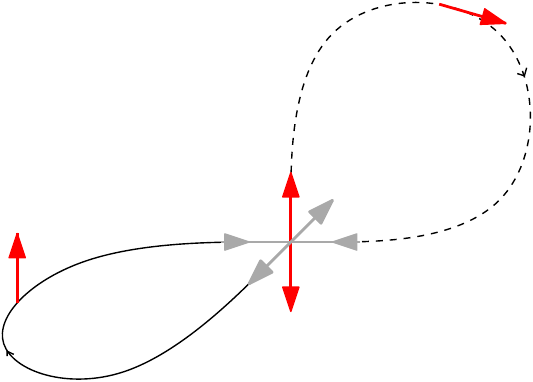}  
\caption{The stable and unstable manifolds of the saddle have 
dimension 1 and 2, respectively. The Perron-Frobinius vector
field is represented in red. The limit set given by the 
homoclinic orbit on the right part of the figure (dashed line) 
is ruled out by Corollary \ref{thm:homoclinic}.
The limit set given by the homoclinic orbit on the 
left part of the figure (solid line) is compatible
with Theorem \ref{thm:limit_sets} but the Perron-Frobenius vector
field is nowhere tangent to the curve.}
\label{fig:homoclinic2}
\end{figure}

\subsection{Complex limit sets of differentially 
positive systems are not attractors.}
\label{sec:complex_attractors}

Part (ii) of Theorem \ref{thm:limit_sets} allows for
more complex limit sets than those described in Part (i), but those
limit sets cannot be attractors, because they are nowhere tangent to the
dominant direction of the linearization.
This property has been well studied for monotone systems.
For instance, Smale proposed a construction to
imbed chaotic behaviors in a cooperative irreducible system
\cite[Chapter 4]{Smith1995}. The transversality of those limit
sets to the Perron-Frobenius vector field 
extends to the trajectories that converge to them.
For instance, consider any $\omega$-limit set $\omega(\xi)$,
$\xi\in \calX$, satisfying Part (ii) of Theorem \ref{thm:limit_sets}.
Any trajectory whose $\omega$-limit points belong to $\omega(\xi)$ 
is nowhere tangent to the Perron-Frobenius vector field.
Moreover, if the trajectory does not converge to a fixed point
then it shows high sensitivity with respect to initial conditions.

\begin{corollary}
\label{thm:unstable_trajectories}
Under the assumptions of Theorem \ref{thm:limit_sets},
suppose that for some $\xi\in\calX$, $\omega(\xi)$
satisfies Part (ii) of Theorem \ref{thm:limit_sets}.
Then, for any $z\in \calX$ such that 
$\omega(z)\subseteq \omega(\xi)$, 
the trajectory $\psi(\cdot,0,z)$ satisfies
$f(\psi(t,0,z)) \notin \calK_\calX(\psi(t,0,z))\setminus\{0\}$
for each $t\geq 0$. 
If $\omega(z)$ is not a singleton, then
$
\liminf\nolimits\limits_{t\to\infty} 
|\partial_z \psi(t,0,z)\mathbf{w}(z)|_{\psi(t,0,z)} = \infty
$.
\end{corollary}

The reason why the possibly complex limit sets of differentially 
positive systems are of little importance
for the overall behavior is that 
their basin of attraction $\calW$ seems strongly repelling.
In accordance to Corollary \ref{thm:unstable_trajectories},
it is very ``likely'' for a trajectory in a small
neighborhood of $\calW$ to move away from $\calW$ 
along the Perron-Frobenius vector field 
and ``unlikely'' to return to $\calW$ at later time.
The argument can be made rigorous for strongly order
preserving monotone systems, allowing
to recover the following celebrated result for 
monotone systems \cite{Smith1995,Hirsch1988}.
\begin{corollary}
\label{thm:almost_convergence}
Let $\Sigma$ be a continuous dynamical system of the form $\dot{x}=f(x)$
on a vector space $\calX$,  strict differentially positive 
with respect to the constant cone field $\calK_\calX \subseteq T_x\calX = \calX$.
Under boundedness of trajectories, 
the $\omega$-limit set $\omega(\xi)$ is a fixed point 
for almost all $\xi\in \calX$.
\end{corollary}
For general differentially positive systems, the above discussion leads to
the following conjecture. 
\begin{conjecture}
\label{conj:limit_sets}
Under the assumptions of Theorem \ref{thm:limit_sets},
for almost every $\xi\in \calX$,
the $\omega$-limit set $\omega(\xi)$ is given by
either a fixed point, or a limit cycle, or 
fixed points and connecting arcs.
\end{conjecture}

The implication of Conjecture \ref{conj:limit_sets}
would be that any limit set not covered by case (i)
in Theorem \ref{thm:limit_sets} could at best
attract a set of initial conditions of zero measure.

\section{Extended example: differential positivity of the damped pendulum}
\label{sec:pendulum} 

The results of the paper are briefly illustrated on 
the analysis of the classical (adimensional) nonlinear pendulum model:
\begin{equation}
\label{eq:pendulum1}
\Sigma:\left\{
\begin{array}{rcl}
\dot{\vartheta} &=& v \\
\dot{v} &=& -\sin(\vartheta) - k v + u
\end{array}
\right. 
\qquad (\vartheta,v) \in \calX := \mathbb{S}\times \real \ ,
\end{equation}
where $k\geq0$ is the damping coefficient and $u$ is the (constant) torque input.

The analysis of the state matrix $A(\vartheta,k)$ for $\vartheta \in \mathbb{S}$
of the variational system
\begin{equation}
\label{eq:pendulum2}
\mymatrix{c}{ \dot{\delta \vartheta} \\ \dot{\delta v} } =
\underbrace{\mymatrix{cc}{ 0 & 1 \\ -\cos(\vartheta) & - k }}_{=:A(\vartheta,k)}
\mymatrix{c}{ \delta \vartheta \\ \delta v } 
\qquad (\delta \vartheta, \delta v) \in T_{(\vartheta,v)}\calX 
\end{equation}
reveals that the pendulum is  strictly differentially positive for
$k>2$ and differentially positive for $k=2$ with respect to the cone field
\begin{equation}
 \calK_\calX(\vartheta,v) := \{(\delta \vartheta, \delta v)\in T_{(\vartheta,v)}\calX \,|\, \delta \vartheta \geq 0, \delta \vartheta+\delta v \geq 0   \}  \ .
\end{equation}

The differential positivity of \eqref{eq:pendulum1} for $k\geq 2$
has the following simple geometric interpretation.
For any $k \geq 2$ and any value of $\vartheta$, the matrix $A(\vartheta,k)$ 
has only real eigenvalues. The blue and the red lines in Figure \ref{fig:pendulum} 
show the direction of the eigenvectors of 
$A(\vartheta,4)$ (left) - $A(\vartheta,3)$ (center) - $A(\vartheta,2)$ (right),
for sampled values of $\vartheta \in \mathbb{S}$.
The blue eigenvectors ($\delta \vartheta \leq 0$)
are related to the smallest eigenvalues, which is negative for each $\vartheta$. 
The red eigenvectors ($\delta \vartheta \geq 0$) are
related to the largest eigenvalues.
$\calK_\calX(\vartheta,v)$ is represented by the shaded area in Figure \ref{fig:pendulum}.
The black arrows represent the  vector field of the variational
dynamics along the boundary of the cone. By continuity and homogeneity
of the vector field on the boundary of the cone,
$\Sigma$ is  strictly differentially positive for each $k > 2$.
It reduces to a differentially positive system in the limit of $k = 2$. 
The loss of contraction in such a case has a simple geometric explanation: 
one of the two eigenvectors of $A(0,2)$ belongs to the boundary of the cone and 
the eigenvalues of $A(0,2)$ are both in $-1$. 
The issues is clear for $u = 0$ at the equilibrium $x_e := (0,0)$. In such a case
$A(0,2)$ gives the linearization of $\Sigma$ at $x_e$ and the 
eigenvalues in  $-1$ makes the positivity of the linearized system non 
strict for any selection of $\calK_\calX(x_e)$.

\begin{figure}[htbp]
\centering
\includegraphics[width=0.32\columnwidth]{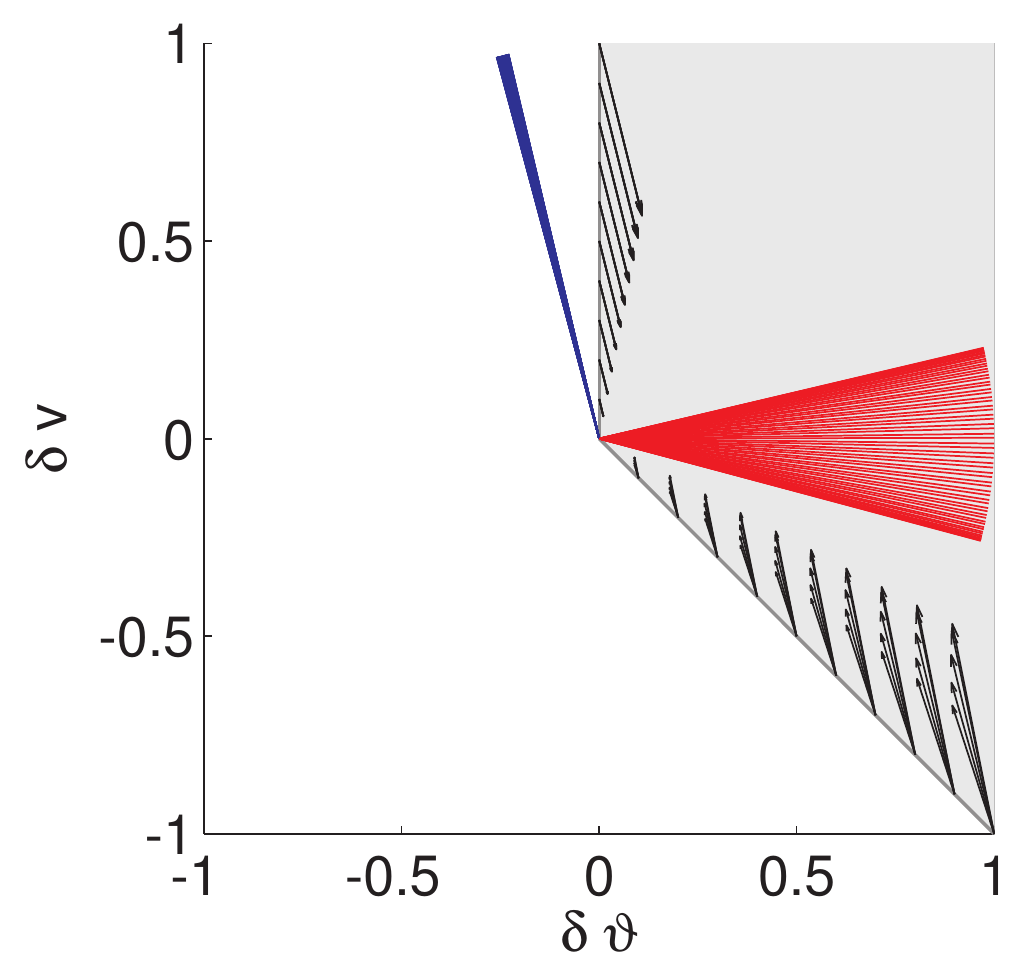}  
\includegraphics[width=0.32\columnwidth]{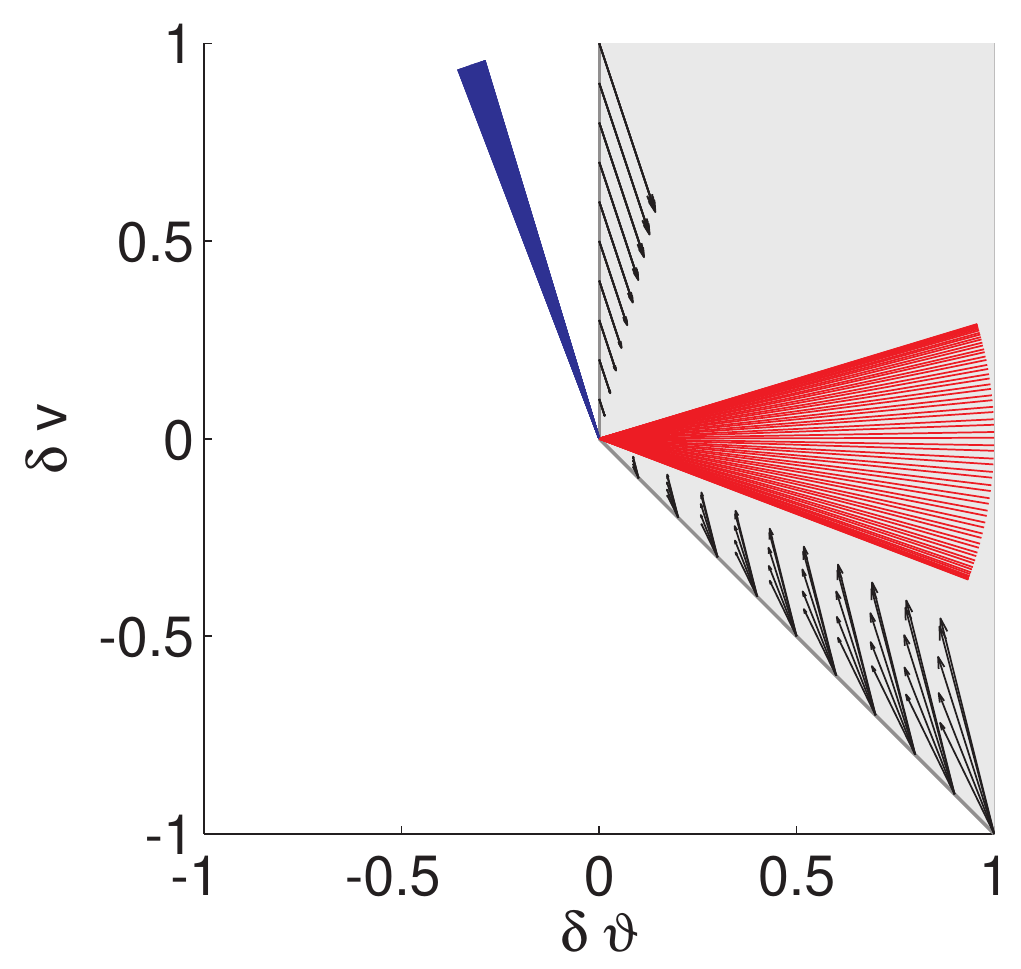}  
\includegraphics[width=0.32\columnwidth]{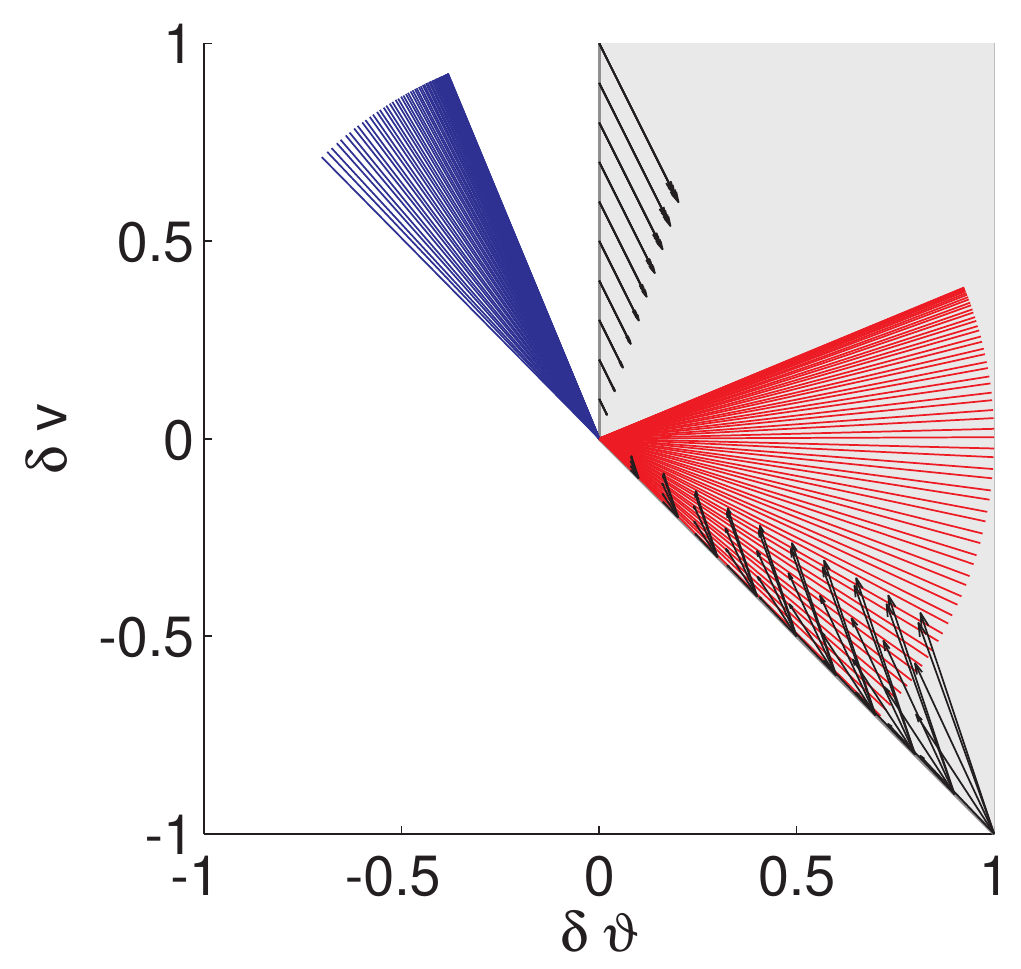}  
\caption{$\calK_\calX(\vartheta,v)$ is shaded in gray. 
The two eigenvector of $A(\vartheta,k)$ for a selection $\vartheta \in \mathbb{S}$
are represented in red and blue.
Red eigenvectors - largest eigenvalues. Blue eigenvectors - smallest eigenvalues.
Left figure: $k=4$. Center figure: $k=3$. Right figure: $k=2$.
}
\label{fig:pendulum} 
\end{figure}

For $k>2$ the trajectories of the pendulum are bounded.
he velocity component converges in finite time to the set
$
\calV := \{v \in\real \,|\, -\rho \frac{|u|+1}{k} \leq v \leq \rho \frac{|u|+1}{k}\}
$,
for any given $\rho > 1$, 
since the kinetic energy $E := \frac{v^2}{2}$ satisfies
$\dot{E} = -kv^2 + v(u - \sin(\vartheta)) \leq (|u|+1-k|v|)|v|  < 0$
for each $|v| > \frac{|u|+1}{k}$. 
The compactness of the set $\mathbb{S}\times \calV$ opens the way
to the use of the results of Section \ref{sec:limit_sets}. 
For $u = 1+ \varepsilon$, $\varepsilon > 0$, 
we have that 
$\dot{v} \geq  \varepsilon - kv$ which, after a transient,
guarantees that 
$\dot{v} > 0$, thus eventually $\dot{\vartheta} > 0$.
Denoting by $f(\vartheta, v)$ the right-hand side in \eqref{eq:pendulum1},
it follows that, after a finite amount of time,
every trajectory belongs to a forward invariant set $\calC \subseteq \mathbb{S}\times \calV$
such that $f(\vartheta,v) \in \mbox{int}\calK_\calX(\vartheta,v)$.
By Corollary \ref{thm:PBtheorem}, there is a unique
attractive limit cycle in $\calC$.

It is of interest to interpret differential positivity against the textbook analysis \cite{Strogatz1994}. 
Following \cite[Chapter 8]{Strogatz1994}, 
Figure \ref{fig:critical_damping} summarizes the 
qualitative behavior of the pendulum for different
values of the damping coefficient $k\geq 0$ and of 
the constant torque input $u\geq 0$
(the behavior of the pendulum for $u\leq0$ is symmetric).
\begin{figure}[bp]
\centering
\includegraphics[width=0.9\columnwidth]{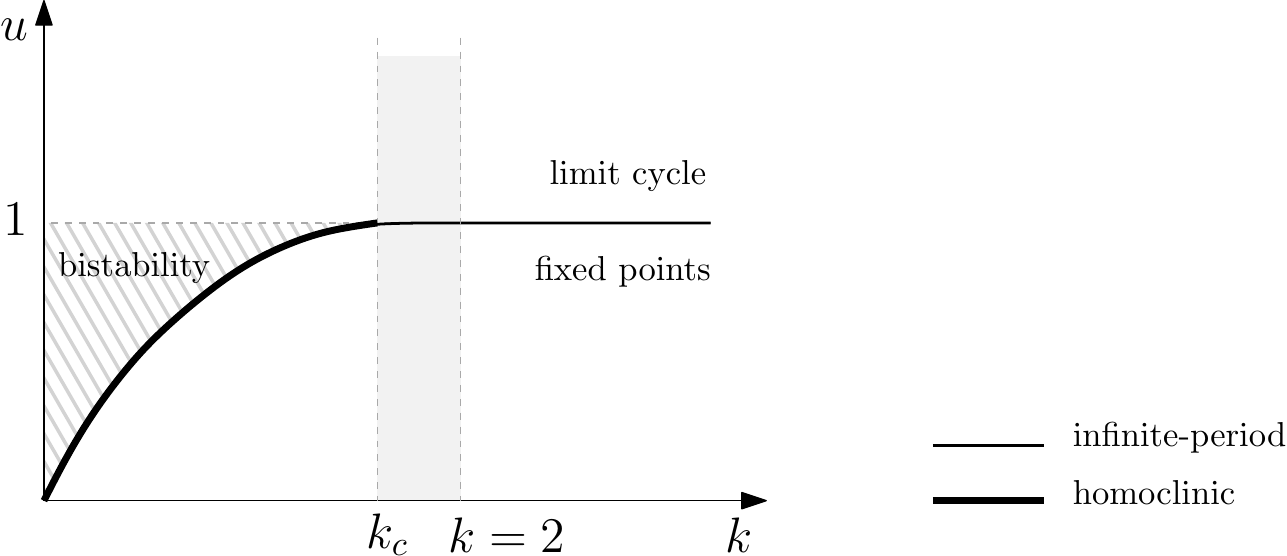}  
\caption{The qualitative behavior of the pendulum for large/small values of torque and damping,
as reported in \cite[p. 272]{Strogatz1994}.}
\label{fig:critical_damping}
\end{figure} 
The nonlinear pendulum cannot be 
differentially positive for arbitrary values of the torque when 
$k\leq k_c$. This is because the region of bistable behaviors
(coexistence of small and large oscillations) is delineated by 
a homoclinic orbit, which is ruled out by differental positivity
(Corollary \ref{thm:homoclinic}).
For instance, looking at Figure \ref{fig:critical_damping},
for any $k< k_c$ there exists a 
value $u = u_c(k)$ for which the pendulum encounters
a \emph{homoclinic bifurcation} 
(see \cite[Section 8.5]{Strogatz1994} and
Figures 8.5.7 - 8.5.8 therein).
In contrast, the \emph{infinite-period bifurcation} 
at $k > 2$, $u=1$ \cite[Chapter 8]{Strogatz1994}
is compatible with differential positivity.

It is plausible that the ``grey'' area between $k_c$ and $k=2$ 
is a region where the nonlinear pendulum is differentially positive
over a uniform time-horizon rather than pointwise. A detailed analysis
of the region $k_c \leq k < 2$ is postponed to a further publication.

\section{Conclusions}

The paper introduces the concept of differential positivity, a local characterization of
monotonicity through the infinitesimal contraction properties of a cone field. The theory
of differential positivity reduces to the theory of monotone systems when the state-space
is linear and when the cone field is constant.
The differential framework allows for a generalization of the Perron-Frobenius theory 
on nonlinear spaces and/or  non constant cone fields.
The paper focuses on the characterization of limit sets of differentially positive systems,
showing that those systems enjoy properties akin to the Poincare-Bendixson theory of
planar systems. In particular, differential positivity is seen as a novel analysis tool for
the analysis of limit cycles and as a property that precludes complex behaviors in
a significant class of nonlinear systems.
Many issues of interest remain to be addressed beyond the material of the present paper.
{The most pressing of those is probably the topic of feedback
interconnections: negative feedback interconnections of monotone systems
are known to provide a key mechanism of oscillation \cite{Gedeon2007,Gedeon2010} and it is appealing to
analyze their differential positivity by inferring a (non-constant) cone field
from the order properties of the subsystems and from the interconnection structure only.
More generally,  the construction of particular cone fields for interconnections of relevance in system theory
(e.g. Lure systems) as well as  the relationship between differential positivity and horizontal
contraction recently studied in \cite{Forni2014} will be the topic of further research.}

\appendix

\subsection{Proofs of Section \ref{sec:Perron-Forbenius}}
\label{sec:PF_proofs}

\begin{proofof}\emph{Theorem \ref{thm:dsch}:}
{
Using $\Gamma_{(x_1,u_1,x_2,u_2)}$ to denote the
linear and invertible mapping 
$\Gamma(x_1,u_1,x_2,u_2)$
that satisfies 
$\Gamma(x_1,u_1,x_2,u_2)\calK_\calX(x_1,u_1) = \calK_\calX(x_2,u_2)$
for each $(x_1,u_1),(x_2,u_2)\in\calX\times\calU$ (see Section \ref{sec:conal_manifold}),
and $d_{*(x,u)}$ to denote $d_{\calK_\calX(x,u)}$ (for readability), 
note that the Hilbert metric satisfies
\begin{equation}
\label{eq:Gamma_invariance}
d_{*(x_2,u_2)}(v,w) = 
 d_{*(x_1,u_1)}(\Gamma_{(x_1,u_1,x_2,u_2)}^{-1}v,\Gamma_{(x_1,u_1,x_2,u_2)}^{-1}w) 
\end{equation}
 for any $(x_1,u_1),(x_2,u_2)\in\calX\times\calU$, and 
any $v,w \in \calK_{\calX}(x_2,u_2)$.
\eqref{eq:Gamma_invariance} follows by the combination of \eqref{eq:Mm} with the 
identity
$\Gamma_{(x_1,u_1,x_2,u_2)}^{-1}\mathrm{bd}\calK_\calX(x_2,u_2)=\mathrm{bd}\calK_\calX(x_1,u_1)$
(by linearity). 

Along any given solution pair $(x(\cdot),u(\cdot)):[t_0,\infty)\to\calX \in \Sigma$ 
define the linear operator 
\begin{equation}
\label{eq:A_definition}
\begin{array}{l}
A_{(x(\cdot),u(\cdot))}[\tau_2,\tau_1] := \vspace{1mm}\\
:= \Gamma_{(x(\tau_1),u(\tau_1),x(\tau_2),u(\tau_2))}^{-1}\partial_{x(\tau_1)}\psi(\tau_2,\tau_1,x(\tau_1),u(\cdot)) 
\end{array}
\end{equation}
where $t_0\leq \tau_1 \leq \tau_2$.
Thus, using $d_{*(t)}$ to denote $d_{\calK_\calX(x(t),u(t))}$,
and recalling that, in Theorem \ref{thm:dsch},
$x(t) = \psi(t,t_0,x(t_0),u(\cdot))$,
$\delta x_1(t) = \partial_{x(t_0)}\psi(t,t_0,x(t_0),u(\cdot))\delta x_1(t_0)$,
and $\delta x_2(t) = \partial_{x(t_0)}\psi(t,t_0,x(t_0),u(\cdot))\delta x_2(t_0)$,
for each $t\geq t_0$ we get 
\begin{equation}
\label{eq:Hilbert_invariance6}
\begin{array}{l}
d_{*(t)}(\delta x_1(t),\delta x_2(t)) \ = \\
= \  d_{*(t_0)}(A_{(x(\cdot),u(\cdot))}[t,t_0]\delta x_1(t_0),A_{(x(\cdot),u(\cdot))}[t,t_0]\delta x_2(t_0))  \\
\leq \  d_{*(t_0)}(\delta x_1(t_0),\delta x_2(t_0))  \ .
\end{array}
\end{equation}
The identity follows by the combination of 
\eqref{eq:Gamma_invariance}, \eqref{eq:A_definition},
and differential positivity.
The inequality follows from the fact that 
$A_{(x(\cdot),u(\cdot))}[t,t_0]$ is a linear operator in $\calK_\calX(x(t_0),u(t_0)) \to \calK_\calX(x(t_0),u(t_0))$,
as in \cite{Pratt1982, Bushell1973}. 

Strict differential positivity guarantees that there exists $T>0$ such
that, for any given $\tau \in [t_0,\infty)$, 
\begin{equation}
\label{eq:A_contraction}
A_{(x(\cdot),u(\cdot))}[\tau+T,\tau] \calK_\calX(x(\tau),u(\tau)) \subseteq \calR_\calX(x(\tau),u(\tau)) \ .
\end{equation}
Thus, following \cite{Bushell1973,Pratt1982}, 
define \emph{projective diameter} 
$\Delta_T := \sup \{ d_{\calK_\calX(x,u)}(v_1,v_2) \,|\, v_1,v_2 \in \calR_\calX (x,u) \} < \infty$
and \emph{contraction ratio} $\mu_T := \tanh \left(\frac{\Delta_T}{4}\right) < 1$.
Then, using again $d_{*(t)}$ to denote $d_{\calK_\calX(x(t),u(t))}$,
\cite[Theorem 3.2]{Bushell1973} and \cite[Proposition 3.14]{Pratt1982}
guarantee the inequality 
\begin{equation}
\begin{array}{l}
d_{*(\tau+T)}(\delta x_1(\tau+T),\delta x_2(\tau+T)) \ = \\
=   d_{*(\tau)}(A_{(x(\!\cdot\!),u(\!\cdot\!))}[\tau\!+\!T,\!\tau]\delta x_1(\tau),A_{(x(\!\cdot\!),u(\!\cdot\!))}[\tau\!+\!T,\!\tau]\delta x_2(\tau))  \\
\leq  \mu_T d_{*(\tau)}(\delta x_1(\tau),\delta x_2(\tau))  \ ,
\end{array}
\end{equation}
for all $\tau \geq t_0$.
By the semigroup property, for any integer $k$, and $t \geq t_0 + kT$ we get
\begin{equation}
\label{eq:Hilbert_invariance11}
d_{*(t)}(\delta x(t),\delta y(t)) 
\leq \mu_T^k d_{*(t_0)}(\delta x(t_0),\delta y(t_0)) 
\end{equation}
which establishes the exponential convergence.

Finally, combining \eqref{eq:Hilbert_invariance11} and \eqref{eq:A_contraction},
for all $t\geq t_0 + (k+1)T$ we get $d_{*(t)}(\delta x(t),\delta y(t)) \leq \mu_T^k \Delta_T$.}
\end{proofof} \vspace{1mm}

\begin{proofof}\emph{Theorem \ref{thm:differential_pf}:}
Consider the solution pair $(z(\cdot),u(\cdot)) \in \Sigma$ such that $z(t) = x$ and define
$
\calC(t,t_0,x,u(\cdot)) := \partial_{z(t_0)}\psi(t,t_0,z(t_0),u(\cdot))\calK_{\calX}(z(t_0),u(t_0))
$.
By construction, strict differential positivity
guarantees that 
\begin{equation}
\label{eq:nonlin_pf_proof2}
\calC(t,t-T_2,x,u(\cdot)) \subseteq \calC(t,t-T_1,x,u(\cdot)) 
\end{equation}
for each $T_2 \geq T_1\geq 0$.
{Also, from Theorem \ref{thm:dsch}
there exists $T>0$ such that 
\begin{equation}
\label{eq:nonlin_pf_proof3}
\limsup\nolimits\limits_{k\to\infty} \{d_{\calK_\calX(x,u(t))}(v_1,\!v_2) \,|\, v_1,\!v_2 \!\in\! \calC(t,t-kT,x,u(\cdot))\} 
= 0,
\end{equation}
since the right-hand side is bounded from above by $\lim_{k\to\infty} \mu_T^{k-1} \Delta_T$,
where $\Delta_T < \infty$ and $\mu_T<1$ are respectively the projective diameter 
and the contraction ratio defined in the proof of Theorem \ref{thm:dsch}.}

Following \cite[Chapter 4, Exercise 4.3]{Rockafellar2004},
the combination of \eqref{eq:nonlin_pf_proof2} and \eqref{eq:nonlin_pf_proof3}
guarantees that 
the set $\calC(t,t-kT,x,u(\cdot))\cap \calB(x)$ converges to the singleton 
$\{\mathbf{w}_{u(\cdot)}(x,t)\}\subseteq \mathrm{int}\calK_\calX(x,u(t)) \cap\calB(x)$
as $k\to \infty$.
We write $\mathbf{w}_{u(\cdot)}(x,t)$ since the limit
of the set $\calC(t,t-kT,x,u(\cdot))\cap \calB(x)$ for $k \to \infty$ 
depends only on
the input signal $u(\cdot)$, the state $x$, and the time $t$.
\end{proofof} \vspace{1mm}

\begin{proofof}\emph{Corollary \ref{cor:differential_time-invariant_pf}:}
To show \emph{time-invariance}, 
consider $u(\cdot) = u$.
Under the action of the constant input $u$, consider the 
trajectories $z(\cdot)$ and $y(\cdot)$ such that $z(t) = x$ and $y(t+T) = x$,
{ for some $t,T\in \real$}.
Using $\calB(x) := \{\delta x \in T_x\calX \,|\, |\delta x|_x=1\}\subseteq T_x\calX$,
from Theorem \ref{thm:differential_pf},
$\lim\nolimits\limits_{t_0\to -\infty}(\partial_{z(t_0)}\psi(t,t_0,z(t_0),\!u)\calK_{\calX}(z(t_0),\!u))\cap\calB(x) \!= \mathbf{w}_u(x,t)$
and $\lim\nolimits\limits_{t_0\to -\infty}(\partial_{y(t_0)}\psi(t+T,t_0,y(t_0),u)\calK_{\calX}(y(t_0),u))\cap\calB(x) = \mathbf{w}_u(x,t+T)$.
However, for constant $u$, $\Sigma$ is time invariant, therefore uniqueness of trajectories from initial conditions 
guarantees that
$z(\tau) = y(\tau+T)$ for all $\tau \in (-\infty,t)$, which guarantees $\calK_{\calX}(z(t_0),u)) = \calK_{\calX}(y(t_0+T),u))$.
$\mathbf{w}_u(x,t) = \mathbf{w}_u(x,t+T)$ follows.

{
To show \emph{continuity}, consider the family of vector fields $\{g_k\}$
given by $g_0(x) \in \calK_\calX(x,u)\cap \calB(x)$
and
$g_{k+1}(\psi(T,0,x,u)) := \frac{\partial_x\psi(T,0,x,u) g_{k}(x)}{|\partial_x\psi(T,0,x,u) g_{k}(x)|_x}$.
Note that 
each $g_k$ is a continuous vector field because $\partial_x\psi(T,0,x,u)$ is continuous in $x$
and the Riemannian structure is smooth in $x$.
Moreover,
$d_{\calK_\calX(x,u)}(g_{k+1}(x),\mathbf{w}_u(x)) \leq \mu_T^{k} \Delta_T$
for all $x \in \calX$ and all $k \geq 1$,
where $\mu_T$ and $\Delta_T$ are respectively the contraction ratio and the projective diameter
defined in the proof of Theorem \ref{thm:dsch}. Indeed, $g_k$ converge uniformly to $\mathbf{w}_u$ 
(with respect to the Hilbert metric at $x$).

By contradiction, suppose that the Perron-Frobenius vector field is not continuous at $x$.
Then, following \cite[Definition 2.1]{boothby2003}, in local coordinates, the
$i$th component $[\mathbf{w}_u(x)]_i$ of $\mathbf{w}_u(x)$ is not a continuous function: 
there exist $\varepsilon > 0$, a sequence of points $y_j\to x$ as $j \to \infty$, and a bound $J$ 
such that the absolute value
$|[\mathbf{w}_u(x)]_i - [\mathbf{w}_u(y_j)]_i| \geq \varepsilon$ for any $j\geq J$.
By the uniform convergence of $g_k$ to $\mathbf{w}_u$, there exists a bound $K$ such that 
$|[g_k(x)]_i - [\mathbf{w}_u(x))]_i| \leq \frac{\varepsilon}{3}$
and
$|[g_k(y_j)]_i - [\mathbf{w}_u(y_j))]_i| \leq \frac{\varepsilon}{3}$ for all $k\geq K$ and all $j$.
Therefore,
 $|[g_k(x)]_i - [g_k(y_j)]_i| \geq \frac{\varepsilon}{3}$ for all $k\geq K$ and $j\geq J$, 
 which contradicts the continuity of $g_k$.
 }
 
Finally, the coincidence between the Perron-Frobenius vector field and 
the Perron-Frobenius vector for \emph{linear systems}
is a straightforward consequence of \eqref{eq:PF_vector_field}.
\end{proofof}

\subsection{Proofs of Section \ref{sec:limit_sets}}
\label{sec:proof_main_theorem}
For readability, in what follows we use
$\psi_t(x):=\psi(t,0,x)$,
$\partial \psi_t(x):=\partial_x \psi(t,0,x)$,
and $d_x(\cdot,\cdot) := d_{\calK_\calX(x)}(\cdot,\cdot)$.
Recall that the pair
$(\psi_t(x),\partial \psi_t(x)\delta x)$ is the trajectory
of the prolonged system $\delta \Sigma$ given by
$\dot{x}=f(x)$, $\dot{\delta x} = \partial f(x) \delta x$
from the initial condition $(x,\delta x) \in T \calX$.

We develop first some technical results.
The claims of the next two lemmas are about the boundedness of the
trajectories of the variational systems. The claims holds for both
continuous and discrete systems (closed or with constant inputs).

{
\begin{lemma}
\label{lem:boundedness1}
Let $u(\cdot) = u$ be constant. Under the assumptions of Theorem \ref{thm:differential_pf}, 
for any $x\in \calX$ and any $\delta x \in T_x \calX$, if
$
|\partial \psi_t(x)\mathbf{w}(x)|_{\psi_t(x)}  < \infty$
then 
$
\limsup\nolimits\limits_{t\to\infty} |\partial \psi_t(x)\delta x|_{\psi_t(x)} < \infty
$.
\end{lemma}
\begin{lemma}
\label{lem:boundedness2}
Let $u(\cdot) = u$ be constant. Under the assumptions of Theorem \ref{thm:differential_pf},
let $x$ be any point of $\calX$ and suppose that there exists 
$\delta x \in \mbox{int}\calK_\calX(x)$ such that
$
\limsup\nolimits\limits_{t\to\infty} |\partial \psi_t(x)\delta x|_{\psi_t(x)}  < \infty
$.
Then, 
$
\limsup\nolimits\limits_{t\to\infty} |\partial \psi_t(x)\mathbf{w}(x)|_{\psi_t(x)} < \infty
$.
\end{lemma} \vspace{1mm}
}

\begin{proofof}\emph{Lemma \ref{lem:boundedness1}:}
For the first item suppose that the implications does not hold
and $ |\partial \psi_t(x)\delta x|_{\psi_t(x)}$ grows unbounded.
Take the vector $\delta y = \delta x + \alpha \mathbf{w}(x)$.
Note that for $\alpha$ sufficiently large $\delta y \in \calK_\calX(x)$. 
{
These facts and the linearity of $\partial \psi_t(x)$ guarantee that
$|\partial \psi_t(x)\delta y|_{\psi_t(x)}$ grows unbounded and 
there exists $\overline{t}$ sufficiently large such that 
either (i) $\partial \psi_{\overline{t}}(x) \delta y \notin \calK_\calX(\psi_{\overline{t}}(x))$,
which contradicts differential positivity,
or (ii) $\partial \psi_{\overline{t}}(x) \delta y \simeq \rho \psi_{\overline{t}}(x) \mathbf{w}(x)$
where $\beta \in \real$ is a scaling factor.
Thus, for all $t\geq \overline{t}$, by linearity, 
$|\partial \psi_{t}(x) \delta y|_{\psi_{t}(x)} 
\simeq \rho |\psi_{t}(x) \mathbf{w}(x)|_{\psi_{t}(x)}$ 
which grows unbounded contradicting the assumption on $|\psi_{t}(x) \mathbf{w}(x)|_{\psi_{t}(x)}$.}
\end{proofof} \vspace{1mm}

\begin{proofof}\emph{Lemma \ref{lem:boundedness2}:}
For the second item,
consider any decomposition 
$\delta x = \alpha \mathbf{w}(x) + \beta \delta z$
where $\alpha,\beta \in \real_{\geq 0}$ and $\delta z \in \calK_\calX(x)$,
which can always be achieved for $\alpha$ sufficiently small
since $\delta x \in \mbox{int}\calK_\calX(x)$. Then, 
$
\partial\psi_t(x)\delta x
=\partial\psi_t(x)[\alpha \mathbf{w}(x) + \beta \delta z]  
= \alpha \partial\psi_t(x) \mathbf{w}(x) + \beta \partial\psi_t(x) \delta z 
$
and, by projective contraction, $\partial\psi_t(x) \delta z$ converges asymptotically to
$\rho_t \mathbf{w}(\psi_t(x))$ for some $\rho_t \in \real_{\geq 0}$. Thus,
\begin{equation}
\begin{array}{l}
\limsup\nolimits\limits_{t\to\infty} |\partial \psi_t(x) \delta x|_{\psi_t(x)} = \\
= \limsup\nolimits\limits_{t\to\infty} |\alpha \partial \psi_t(x) \mathbf{w}(x) + \beta \rho_t \mathbf{w}(\psi_t(x))|_{\psi_t(x)} \\ 
= \limsup\nolimits\limits_{t\to\infty} \left|\alpha \partial \psi_t(x) \mathbf{w}(x) + \beta \rho_t \frac{\partial\psi_t(x) \mathbf{w}(x)}{|\partial\psi_t(x) \mathbf{w}(x)|_{\psi_t(x)}}\right|_{\psi_t(x)} \\ 
=  \limsup\nolimits\limits_{t\to\infty}  \left(\alpha+  \frac{\beta\rho_t}{|\partial\psi_t(x) \mathbf{w}(x)|_{\psi_t(x)}}\right) |\partial\psi_t(x) \mathbf{w}(x)|_{\psi_t(x)} \\
\geq \alpha \limsup\nolimits\limits_{t\to\infty} |\partial\psi_t(x) \mathbf{w}(x)|_{\psi_t(x)} \ . \vspace{-5mm}
\end{array} 
\end{equation}
\end{proofof}

The next lemma shows that any
trajectory of a continuous and closed differentially positive system
whose motion follows the Perron-Frobenius vector field
either converges to a fixed point or 
defines a periodic orbit.
In what follows we will use $\psi_t(x):=\psi(t,0,x)$
and $\partial \psi_t(x):=\partial_x \psi(t,0,x)$.

\begin{lemma}
\label{lem:Hirsch}
Under the assumptions of Theorem \ref{thm:limit_sets},
consider any $x$ such that,
for all $t$, $f(\psi_t(x)) = \lambda(\psi_t(x)) \mathbf{w}(\psi_t(x))$
and $|\lambda(\psi_t(x))| \geq \rho>0$.
Then, the trajectory $\psi_t(x)$ is periodic.
\end{lemma}
\begin{proofof}\emph{Lemma \ref{lem:Hirsch}:}
In what follows we use $\calA:=\{\psi_t(x)\,|\, t\in\real\}$ and 
$\calB_\varepsilon(x)$ to denote
a ball of radius $\varepsilon$ centered at $x$:
for any two points $z$ 
in $\calB_\varepsilon(x)$ there exists a curve $\gamma(\cdot)$ such that
$\gamma(0) = x$, $\gamma(1) = z$ and whose length
$L(\gamma(\cdot)) = \int_0^1 |\dot{\gamma}(s)|_{\gamma(s)}ds \leq \varepsilon$.
We make also use of the notion
of \emph{local section} at 
$x$, which is any open set $\calS\subseteq\calX$ 
of dimension $n-1$
($\calX$ has dimension $n$)
contained within a (sufficiently) small neighborhood 
$\calB_\varepsilon(x)$ of $x$ such that
$x \in \calS$ and $\mathbf{w}(z) \notin T_z\calS$ 
for each $z\in \calS$. 
Finally, for any given Rimannian tensor such that 
$\langle \delta x, \mathbf{w}(x) \rangle_x \geq 0$ for any $x\in \calX$ and $\delta x \in \calK_\calX(x)$,
define the vertical projection
$W_x(\delta x) := \langle \delta x, \mathbf{w}(x)\rangle_x \mathbf{w}(x)$, 
and the horizontal projection
$H_x(\delta x) := \delta x - W_x(\delta x)$. 

\textbf{1)} \underline{Bounded variational dynamics}: 
$0\neq f(\psi_t(x)) = \lambda(\psi_t(x))\mathbf{w}(\psi_t(x))$ for all $t\geq 0$ therefore,
by continuity of the vector field and boundedness of trajectories,
for $\varepsilon > 0$ sufficiently small,
$f(z) \in \mbox{int}\calK_\calX(z)$ or
$-f(z) \in \mbox{int}\calK_\calX(z)$ 
for all $z \in \calE := \bigcup_{t\geq 0}\calB_\varepsilon(\psi_t(x))$.
Note that $\calE$ is a compact set by boundedness of trajectories.
Without loss of generality consider
$f(z) \in \mbox{int}\calK_\calX(z)$. Then,
$f(\psi_t(z)) \in \mbox{int}\calK(\psi_t(z))$
for all $z \in \calE$ and $t\geq 0$,
by differential positivity combined with
the identity $f(\psi_t(z)) = \partial \psi_t(z)f(z)$, which
makes { $(\psi_t(z),f(\psi_t(z)))$ a trajectory of the prolonged system.} 

{
By boundedness of trajectories, 
$|\partial \psi_t(z)f(z)|_{\psi_t(z)}=|f(\psi_t(z))|_{\psi_t{z}}$ is necessarily bounded.}
Lemma \ref{lem:boundedness2} guarantees that
$\limsup\nolimits\limits_{t\to\infty} |\partial \psi_t(z)\mathbf{w}(z)|_{\psi_t(z)} \!\!<\! \infty$
for all $z \in \calE$.
By Lemma \ref{lem:boundedness1}, 
$\limsup\nolimits\limits_{t\to\infty} |\partial \psi_t(z)\delta z|_{\psi_t(z)} \!\!<\! \infty$
for all $z\!\in\!\calE$ and $\delta z \!\in\! T_z\calX$.
Similar results can be obtained for the case
$-f(z) \in \mbox{int}\calK_\calX(z)$ exploiting the linearity
of $\partial\psi_t(z)$.
Finally,
consider any $z\in \calE$ and define 
$
\alpha_z := \sup\nolimits\limits_{\delta z \in T_z\calX,|\delta z|_z = 1} 
\!\limsup\nolimits\limits_{t\to\infty} |\partial \psi_t(z)\delta z|_{\psi_t(z)} < \infty 
$.
Then, for any $\delta z \in T_z\calX$,
$\limsup\nolimits\limits_{t\to\infty} |\partial \psi_t(z)\delta z|_{\psi_t(z)} < \alpha_z |\delta z|_z$.
Since $\calE$ is a compact set, there exists 
$\overline{\alpha} := \sup_{z\in \calE} \alpha_z$. 

\textbf{2)} \underline{Contraction of the horizontal component}:
Take $z\in \calE$.
For $\delta z \in \calK_\calX(z)$, combining the contraction property
$\lim_{t\to\infty}
d_{\psi_t(z)} (\partial \psi_t(z) \delta z, \mathbf{w}(\psi_t(z))) = 0$ 
{ of
Theorem \ref{thm:dsch} and
the bound 
$\limsup\nolimits\limits_{t\to\infty} |\partial \psi_t(z)\delta z|_{\psi_t(z)} \leq \overline{\alpha} |\delta z|_z$
in \textbf{1)}, we get}
$\lim\nolimits\limits_{t\to\infty} \partial \psi_t(z) \delta z - W_{\psi_t(z)}( \partial \psi_t(z)  \delta z) = 0$,
that is, $\lim\nolimits\limits_{t\to\infty}  H_{\psi_t(z)}(\partial \psi_t(z) \delta z) = 0$.
A similar result holds for $\delta z \in -\calK_\calX(z)$.
Consider now $\delta z \notin\calK_\calX(z)$. 
Define the new 
vector $\delta z^* = \delta z + \alpha \mathbf{w}(z)$. For $\alpha$
sufficiently large $\delta z^* \in \calK_\calX(z)$. Therefore
$\lim\nolimits\limits_{t\to\infty}  H_{\psi_t(z)}(\partial \psi_t(z) \delta z^*) = 0$ which implies 
$\lim\nolimits\limits_{t\to\infty} H_{\psi_t(z)}(\partial \psi_t(z) \delta z) = 0$.

\textbf{3)} \underline{Attractiveness of $\psi_t(x)$}:
{Consider the case $t=0$ since $\psi_{0}(x) = x$ 
(the argument is the same for $t> 0$) and take} 
any curve $\gamma(\cdot):[0,1]\to \calE$ such that 
$\gamma(0) =x$ and $L(\gamma(\cdot)) = \varepsilon$,
and consider the evolution of $\gamma(\cdot)$ along the flow 
of the system, that is, $\psi_t(\gamma(s))$ for $s\in [0,1]$.
We observe that
$\frac{d}{ds} \psi_t(\gamma(s)) \!=\! [\partial \psi_t(\gamma(s))] \dot{\gamma}(s)$.
Thus, by \textbf{1)},
$
\limsup\nolimits\limits_{t\to\infty}
|\frac{d}{ds} \psi_t(\gamma(s))|_{\psi_t(\gamma(s))} \leq 
\overline{\alpha} |\dot{\gamma}(s)|_{\gamma(s)}
$, {which guarantees} $\limsup\nolimits\limits_{t\to\infty}
L(\psi_t(\gamma(\cdot))) \leq \varepsilon\overline{\alpha}$.

{
By \textbf{2)}, $\lim\nolimits\limits_{t\to\infty} 
H_{\psi_t(\gamma(s))}( \frac{d}{ds} \psi_t (\gamma(s)))|_{\psi_t(\gamma(s))}
 = 0$
for all $s\in [0,1]$. Thus, 
$\frac{d}{ds}\psi_t(\gamma(s))$ either converges to zero or aligns to 
the Perron-Frobenius vector field. Precisely, three cases may occur:
\begin{itemize}
\item
$\lim\nolimits\limits_{t\to\infty} 
\frac{d}{ds} \psi_t (\gamma(s))) = 0$,
\item
$\lim\nolimits\limits_{t\to\infty} d_{\psi_t(\gamma(s))} ( \mathbf{w}(\psi_t(\gamma(s))) , \frac{d}{ds}\psi_t(\gamma(s)) ) = 0$,
\item
$\lim\nolimits\limits_{t\to\infty} d_{\psi_t(\gamma(s))} ( \mathbf{w}(\psi_t(\gamma(s))) , -\frac{d}{ds}\psi_t(\gamma(s)) ) = 0$.
\end{itemize}
Furthrmore, 
$\lim\nolimits\limits_{t\to\infty} d_{\psi_t(\gamma(s))} ( \mathbf{w}(\psi_t(\gamma(s))) , f(\psi_t(\gamma(s))) ) = 0$
since $f(\psi_t(\gamma(s))) \in \calK_\calX(\psi_t(\gamma(s)))$.
Thus, in the limit, the image of $\psi_t(\gamma(\cdot))$ 
is given by the image of a (time-dependent Perron-Frobenius) curve $\gamma^{\mathbf{w}}_t(\cdot)$ that satisfies
either $\frac{d}{ds}\gamma^\mathbf{w}_t(s) = \mathbf{w}(\gamma^\mathbf{w}_t(s))$
or $\frac{d}{ds}\gamma^\mathbf{w}_t(s) = -\mathbf{w}(\gamma^\mathbf{w}_t(s))$
at any fixed $t$.
By construction,
$\frac{d}{ds}\gamma^\mathbf{w}_t(s) = \frac{1}{\lambda(\gamma^\mathbf{w}_t(s))} f(\gamma^\mathbf{w}_t(s))$.

At each fixed $t$, $\gamma^{\mathbf{w}}_t(\cdot)$ is a (reparameterized) integral curve of the vector field $f$
from the initial condition $\gamma^{\mathbf{w}}_t(0) = \psi_t(x)$. Therefore, 
trajectory with initial condition in} $\gamma(s)$ converges
asymptotically to $\calA$, for all $s\in[0,1]$. 
In particular, { using the bound 
$L(\psi_t(\gamma(\cdot))) \leq \varepsilon\overline{\alpha}$ characterized above},
each trajectory converges asymptotically to
{
\begin{equation}
\label{eq:wow1}
\calC_t(x):= \left\{\psi_{t+\tau}(x)\in \calA\,|\, 
-\frac{\overline{\alpha}\varepsilon}{\rho}  \leq \tau \leq  \frac{\overline{\alpha}\varepsilon}{\rho}  \right\} \ .
\end{equation}
}
\textbf{4)} \underline{Periodicity of the orbit}:
$\psi_t(x)$ does not converge to a fixed point and belongs to a compact set for each $t$, 
therefore there exists a point $\psi_{t_0}(x)$ whose neighborhood $\calB_{\varepsilon}(\psi_{t_0}(x))$
is visited by the trajectory infinitely many times for any given $\varepsilon>0$.
{
For simplicity, without any loss of generality, we consider this point given at $t_0 = 0$, 
that is, $\psi_{0}(x) = x$.}

Consider a local section $\calS$ at $x$
and consider the sequence $t_k\to\infty$ such that $\psi_{t_k}(x)\in \calS$.
{
Since $f(x)$ is aligned with $\mathbf{w}(x)$, 
for $\varepsilon$ sufficiently small, 
the continuity of the system vector field $f$
guarantees that} 
$\calS$ is transverse to $f(z)$ for all $z \in \calS\cap\calB_\varepsilon(x)$, 
that is, $f(z) \notin T_z\calS\cap\calB_\varepsilon(x)$.

By {\textbf{3)}}, for every $z \in \calS\cap\calB_\varepsilon(x)$, 
$\psi_{t_k}(z)$ converges asymptotically to the set $\calC_{t_k}(x)$ as $k\to\infty$.
{For any positive integer $N$, define the ($\frac{\varepsilon}{N}$-inflated) set 
\begin{equation}
\label{eq:wow2}
\begin{array}{l}
\calC_{t}^{(\varepsilon/N)}(x) := \{ y\,|\, \exists \gamma:[0,1]\to \calE, \exists s \in [0,1] \mbox{ such that } \vspace{0.5mm}\\
\hspace{26mm}\gamma(0)\in \calC_{t}(x),\, 
 \gamma(s)=y, 
 L(\gamma(\cdot)) \leq \varepsilon/N\} \, . \vspace{1mm}
\end{array}
\end{equation}
}
Then, by continuity, for every $N>0$, there exists a $k\geq k_N$ sufficiently large such that
$
 \psi_{t_{k}}(z) \in  \calC_{t_{k}}^{(\varepsilon/N)}(x) 
$
for all  $z \in \calS\cap\calB_\varepsilon(x)$.

By the transversality of the section $\calS$ with respect to the system vector field,
{ for $N$ sufficiently large},
we have that {the flow} from $z\in \calC_{t_{k}}^{(\varepsilon/N)}(x)$ 
satisfies $\psi_{\tau}(z) \in \calS$ for some 
$-(\frac{\overline{\alpha}}{\rho}+\delta_N) \varepsilon \leq \tau\leq (\frac{\overline{\alpha}}{\rho}+\delta_N)\varepsilon$,
where $\delta_N$ is some (small) positive constant such that $\delta_N \to 0$ as $N\to \infty$. Moreover,
by continuity with respect to initial conditions, 
for $N$ sufficiently large, we get
{
\begin{equation}
\label{eq:wow3}
\psi_{\tau}(z) \in \calS \cap \calB_{\frac{\varepsilon}{3}}(\psi_{t_k}(x)) \ .
\end{equation}
It follows that, for $t_k -(\frac{\overline{\alpha}}{\rho}+\delta_N) \varepsilon \leq t \leq t_k+(\frac{\overline{\alpha}}{\rho}+\delta_N)\varepsilon$,}
the flow $\psi_t(\cdot)$ maps every point of $\calS\cap\calB_\varepsilon(x)$ 
into $\calS \cap \calB_{\frac{\varepsilon}{3}}(\psi_{t_k}(x))$. 

{For $k \geq k_N$, denote by $P_{k}$ the mapping
from $\calS\cap\calB_\varepsilon(x)$ into $\calS \cap\calB_{\frac{\varepsilon}{3}}(\psi_{t_k}(x))$}.
Since $\psi_{t_k}(x)$ recursively visit any local section of $x$,
eventually, for some $K\geq k_N$, the flow satisfies 
$\psi_{t_K}(x) \in \calS\cap \calB_{\frac{\varepsilon}{3}}(x)$.
{
Using the results above, we conclude that 
the flow of the system maps $\calS\cap\calB_\varepsilon(x)$ 
into $\calS \cap \calB_{\frac{\varepsilon}{3}}(\psi_{t_{K}}(x))\subseteq \calS \cap\calB_{\frac{2\varepsilon}{3}}(x)$,
that is, $P_K$ is a contraction.
By Banach fixed-point theorem
$P_{K}(x) = x$, that is, $\psi_{t_K}(x)=x$.}
\end{proofof} \vspace{1mm}

We are now ready for the proof of the main theorem. \\
\begin{proofof}\emph{Theorem \ref{thm:limit_sets}:}
For any $\xi \in \calX$ consider $\omega(\xi)$.
Three cases may occur: 

\noindent\textbf{1)} $f(x)=0$ for some $x\in \omega(\xi)$.
$x$ is a fixed point.

\noindent\textbf{2)} 
$f(x) \in \mbox{int}\calK_\calX(x)\setminus\{0\}$ or $-f(x) \in \mbox{int}\calK_\calX(x)\setminus\{0\}$ for some $x\in \omega(\xi)$.
In such a case,
\begin{equation}
\label{eq:invariance}
f(z) = \lambda(z) \mathbf{w}(z) \quad \mbox{for all } z\in \omega(\xi) \ ,
\end{equation}
where $\lambda(z) \in \real$ is a scaling factor.
To see this, consider the case $f(x) \in \mbox{int}\calK_\calX(x)$ (wlog).
By definition of $\omega$-limit set, 
there exists a sequence $t_k\to\infty$ such that
$\lim\nolimits\limits_{k\to\infty} \psi_{t_k}(\xi) = x$. For $k\geq k^*$ sufficiently large, $\psi_{t_k}(\xi)$ belongs to an infinitesimal
neighborhood of $x$ therefore $f(\psi_{t_k}(\xi)) \in \calK(\psi_{t_k}(\xi)) $
by continuity of the cone field since $f(x) \in \mbox{int}\calK_\calX(x)$.
Then, by projective contraction, 
$\lim\nolimits\limits_{k\to\infty} d_{\psi_{t_k}(\xi)}(f(\psi_{t_k}(\xi)), w(\psi_{t_k}(x)))=0$, that is, 
$f(x) = \lambda(x) \mathbf{w}(x)$
for some scaling factor $\lambda(x)\in\real$.
By definition of $\omega$-limit set, 
starting from $t_{k^*}$, it is possible to find a 
sequence $\tau_k\to\infty$ as $k\to\infty$ such that
$\lim\nolimits\limits_{k\to\infty}\psi_{t_{k^*}+\tau_k}(\xi) = z$
for any $z\in \omega(\xi)$. Thus, by the argument above,
$f(z) = \lambda(z) \mathbf{w}(z)$ for all $z\in \omega(\xi)$.
\eqref{eq:invariance} guarantees 
that, for any $x\in \omega(\xi)$,
the image of the trajectory $\psi_t(x)$ 
is a subset of the image of some Perron-Frobenius curve $\gamma^{\mathbf{w}}(\cdot)$.
Note that $\lambda(\psi_{t}(x))$ may converge to zero. 
In such a case $\psi_t(x)$ converges to a fixed point.
Otherwise, $|\lambda(\psi_{t}(x))|\geq \varepsilon>0$ therefore,
by Lemma \ref{lem:Hirsch}, $\psi_t(x)$ is periodic.

\noindent\textbf{3)}
It remains to consider the case  
$f(x) \notin \calK_\calX(x)$ 
(or $-f(x) \notin \calK_\calX(x)$) for some $x\in \omega(\xi)$.
In such a case, from the previous item, 
$f(z) \notin \calK_\calX(z)\setminus\{0\}$ for all $z\in \omega(\xi)$. 
Then, either
$\lim_{t\to\infty} f(\psi_t(x)) = 0 $ 
($\psi_t(x)$ converges to a fixed point)
or the contraction of the Hilbert metric enforces
$
\liminf\nolimits\limits_{t\to\infty} |\partial \psi_{t}(x) \mathbf{w}(x)|_{\psi_t(x)} = \infty 
$.
For the latter, consider any sequence $t_k\to\infty$ as $k\to\infty$
such that $f(\psi_{t_k}(x)) \geq \varepsilon>0$.
Take $\delta x = f(x) + \lambda \mathbf{w}(x)$. 
For $\lambda$ sufficiently large $\delta x\in \calK_\calX(x)$.
Then, 
$\lim\nolimits\limits_{k\to\infty} d_{\psi_{t_k}(x)}(\partial \psi_{t_k}(x) \delta x, w(\psi_{t_k}(x))) = 0$
holds only if the evolution of $\delta x$ along the flow
$\partial \psi_{t_k}(x)\delta x 
= \partial \psi_{t_k}(x)f(x) + \lambda \partial \psi_{t_k}(x)\mathbf{w}(x) 
= f(\psi_{t_k}(x)) + \lambda \partial \psi_{t_k}(x)\mathbf{w}(x)$
shows an unbounded growth of the component $\partial \psi_{t_k}(x)\mathbf{w}(x)$.
\end{proofof} \vspace{1mm}

\begin{proofof}\emph{Corollary \ref{thm:PBtheorem}:} 
{
Recall that $f(\psi_t(x)) = \partial \psi_t(x)f(x)$. 
Since $f(x) \in \mathrm{int}\calK_\calX(x)$, 
Lemmas \ref{lem:boundedness1} and \ref{lem:boundedness2} guarantee that
$\limsup\nolimits\limits_{t\to\infty} |\partial \psi_{t}(x) \mathbf{w}(x)|_{\psi_t(x)} < \infty $. 
Since $f(x)\neq 0$ in $\calC$, we conclude that} Case (ii) of Theorem \ref{thm:limit_sets} does not occur.
{Exploiting again the assumption $f(x)\neq 0$ in $\calC$, Case (i)} 
of Theorem \ref{thm:limit_sets} guarantees that the trajectories of 
$\Sigma$ converge to periodic orbits. { We need to prove
the uniqueness.}

By contradiction,
suppose that $\calA_1$ and $\calA_2$ are two periodic
orbits such that $\calA_1 \cap \calA_2 = \emptyset$. 
Take any curve  $\gamma(\cdot):[0,1]\to \calC$
such that $\gamma(0) \in \calA_1$ and
$\gamma(1)\in \calA_2$, and recall that 
$\frac{d}{ds}\psi_t(\gamma(s)) = [\partial\psi_t(x)_{x=\gamma(s)}]\dot{\gamma}(s)$.
{Since $f(\psi_t(x)) = \partial \psi_t(x)f(x)$ and $f(x) \in \mathrm{int}\calK_\calX(x)$,
Lemmas \ref{lem:boundedness1} and \ref{lem:boundedness2} guarantee that}
$\limsup\nolimits\limits_{t\to\infty} |\frac{d}{ds}
\psi_t(\gamma(s))|_{\psi_t(\gamma(s))} < \infty$.
{Since $f(x) \in \mathrm{int}\calK_\calX(x)$ for any $x\in \calC$ and 
the trajectories of $\Sigma$ are bounded, we can use the argument
in \textbf{2)} and \textbf{3)} of Lemma \ref{lem:Hirsch} to show that}
$\frac{d}{ds}\psi_t(\gamma(s))$ converges asymptotically
to $\lambda_a(\psi_t(\gamma(s))) \mathbf{w}(\psi_t(\gamma(s)))$,
thus to $\lambda_b(\psi_t(\gamma(s))) f(\psi_t(\gamma(s)))$,
for some (bounded) scaling factors 
$\lambda_a(\cdot), \lambda_b(\cdot) \in \real$,

As a consequence, {every trajectories whose initial 
conditions belongs to the image of $\gamma(\cdot)$
converges asymptotically to an integral curve of
the system vector field $f(x)$, for $x\in \calC$,
connecting $\calA_1$ and $\calA_2$,}
since $\psi_t(\gamma(0))\in \calA_1$ and 
$\psi_t(\gamma(1))\in \calA_2$ for all $t\geq 0$.
{ It follows that $\calA_1\cap\calA_2 \neq \emptyset$}. A contradiction.
\end{proofof} \vspace{1mm}

\begin{proofof}\emph{Corollary \ref{thm:homoclinic}:}
For some $x\in \calX$, suppose that
$y_e = \lim\nolimits\limits_{t\to\infty} \psi_{-t}(x)$
and $z_e = \lim\nolimits\limits_{t\to\infty} \psi_{t}(x)$
are hyperbolic fixed point.

Suppose that the orbit connecting $y_e$ to $z_e$
is tangential to $\mathbf{w}(y_e)$ at $y_e$.
Take now any point $y$ in a small neighborhood
of $y_e$ such that $y = \psi_{-T}(x)$ for
some $T>0$. By continuity, 
$f(y) \in \mbox{int} \calK_\calX(y)$. Thus,
$\lim\nolimits\limits_{t\to\infty} d_{\psi_t(y)}(
f(\psi_t(y)),\mathbf{w}(\psi_r(y)))=
\lim\nolimits\limits_{t\to\infty} d_{\psi_t(y)}(
f(\psi_t(x)),\mathbf{w}(x_e)) = 0$, 
by Theorem \ref{thm:dsch}.
\end{proofof} \vspace{1mm}

\begin{proofof}\emph{Corollary \ref{thm:unstable_trajectories}:} 
Consider the trajectory $\psi_t(z)$.
Following the proof of Corollary \ref{thm:homoclinic},
necessarily, $f(\psi_t(z)) \notin \calK_\calX(\psi_t(z))$
for any $t\geq 0$. For instance, by contradiction,
suppose that $f(\psi_t(z)) \in \calK_\calX(\psi_t(z))$
for some $t\geq 0$. 
By definition, there exists a sequence
$t_k\to\infty$ as $k\to\infty$ such that 
$\lim\nolimits\limits_{k\to\infty}\psi_{t+t_k}(z) = x\in\omega(\xi)$
thus $\partial \psi_{t+t_k}(z) f(z) = f(x) \notin \calK_\calX(x)$.
By continuity, since $\calK_\calX(x)$ is closed, there exists 
$k^*$ sufficiently large $\partial \psi_{t+t_{k^*}}(z) f(\xi) \notin \calK_\calX(\psi_{t+t_{k^*}}(z))$. But this contradicts differential positivity.

Suppose now that $\omega(z)\subseteq\omega(\xi)$ is not a fixed point.
Then, there exists a sequence
$t_k\to\infty$ as $k\to\infty$ such that 
$\lim_{k\to\infty}f(\psi_{t_k}(z)) = f(x)\neq 0$ for some $x\in \omega(\xi)$.
Take $\delta z = f(z) + \lambda \mathbf{w}(z)$. 
For $\lambda$ sufficiently large $\delta z\in \calK_\calX(z)$.
Then, 
$\lim\nolimits\limits_{t_k\to\infty} 
d_{\psi_{t_k}(z)}(\partial \psi_{t_k}(z) \delta z, w(\psi_{t_k}(z))) = 0$
holds only if the evolution of $\delta z$ along the flow
$\partial \psi_{t_k}(z)\delta z
= \partial \psi_{t_k}(z)f(z) + \lambda \partial \psi_{t_k}(z)\mathbf{w}(z) 
= f(x) + \lambda \partial \psi_{t_k}(z)\mathbf{w}(z)$
shows an unbounded growth of the component 
$\partial \psi_{t_k}(z)\mathbf{w}(z)$.
\end{proofof}

\begin{proofof}\emph{Corollary \ref{thm:almost_convergence}:}
Consider Part (i) of Theorem \ref{thm:limit_sets}.
For any $x\in \omega(\xi)$, we have 
$f(\psi_t(x))=\lambda(\psi_t(x))\mathbf{w}(\psi_t(x))$.
{On vector spaces, for constant cone fields,}
closed curves cannot occur because 
every Perron-Frobenius curve is open. 
Therefore, $\lim\nolimits\limits_{t\to\infty}|\lambda(\psi_t(x))| =0$
by boundedness of solutions.
Consider Part (ii) of Theorem \ref{thm:limit_sets}
and take any $x\in \omega(\xi)$.
Either $\lim\nolimits\limits_{t\to\infty} f(\psi_t(x)) = 0$,
thus $\psi_t(x)$ converges to a fixed point for $t\to\infty$,
or 
$\liminf\nolimits\limits_{t\to\infty}|\partial\psi_t(x)\mathbf{w}(x)|_{\psi_t(x)}=\infty$.
 
{ This last case covers attractors which are not fixed points.
We show that their basin of attraction has dimension $n-1$ at most. 
By contradiction, let $\calA$ be an attractor with a basin of attraction
$\calB_\calA$ of dimension $n$. 
By Corollary \ref{thm:unstable_trajectories},
from every $x\in \calB_\calA$, 
$\liminf\nolimits\limits_{t\to\infty}|\partial\psi_t(x)\mathbf{w}(x)|_{\psi_t(x)}=\infty$.
}
Let $\gamma^{\mathbf{w}}(\cdot)$ be any Perron-Frobenius curve 
such that $\gamma^{\mathbf{w}}(0)=x$.
{
Since $\calB_\calA$ has dimension $n$, there exists
an interval $[\underline{s},\overline{s}] \ni 0$ such that
$\gamma^{\mathbf{w}}(s) \in \calB_\calA$ for all $s\in [\underline{s},\overline{s}]$.
Also, 
$\liminf\nolimits\limits_{t\to\infty}|[\partial\psi_t(x)_{x=\gamma(s)}]\dot{\gamma}^{\mathbf{w}}(s)|_{\psi_t(x)}=\infty$
for all $s\in [\underline{s},\overline{s}]$, that is,
$\liminf\nolimits\limits_{t\to\infty} 
L(\psi_t(\gamma^{\mathbf{w}}(\cdot))) = \infty$.

For each $t > 0$, 
the curve $\gamma_t(\cdot) := \psi_t(\gamma^{\mathbf{w}}(\cdot))$
is a reparameterization of a Perron-Frobenius curve, that is,
$\frac{d}{ds} \gamma_t(s) = \lambda_t(s) w(\gamma_t(s))$
where $\lambda_t(s)$ is a scalar.
Thus, $\gamma_t(\cdot)$ is an open curve for each $t$
that grows unbounded as $t\to \infty$.
It follow that, for all $s\in [\underline{s},\overline{s}]$, 
the trajectory $\psi_t(\gamma^{\mathbf{w}}(s))$ 
grows unbounded},
contradicting the assumption on the 
boundedness of the trajectories of $\Sigma$. 
\end{proofof}

 \bibliographystyle{plain}

\begin{thebibliography}{10}

\bibitem{Absil2008}
P.A. Absil, R.~Mahony, and R.~Sepulchre.
\newblock {\em Optimization Algorithms on Matrix Manifolds}.
\newblock Princeton University Press, Princeton, NJ, 2008.

\bibitem{Aliluiko2006}
A.M. Aliluiko and O.H. Mazko.
\newblock Invariant cones and stability of linear dynamical systems.
\newblock {\em Ukrainian Mathematical Journal}, 58:1635--1655, 2006.

\bibitem{Angeli2003}
D.~Angeli and E.D. Sontag.
\newblock Monotone control systems.
\newblock {\em IEEE Transactions on Automatic Control}, 48(10):1684 -- 1698,
  2003.

\bibitem{Angeli2004b}
D.~Angeli and E.D. Sontag.
\newblock Interconnections of monotone systems with steady-state
  characteristics.
\newblock In MarcioS. Queiroz, Michael Malisoff, and Peter Wolenski, editors,
  {\em Optimal Control, Stabilization and Nonsmooth Analysis}, volume 301 of
  {\em Lecture Notes in Control and Information Science}, pages 135--154.
  Springer Berlin Heidelberg, 2004.

\bibitem{Angeli2004a}
D.~Angeli and E.D. Sontag.
\newblock Multi-stability in monotone input/output systems.
\newblock {\em Systems \& Control Letters}, 51(3--4):185 -- 202, 2004.

\bibitem{Angeli2008}
D.~Angeli and E.D. Sontag.
\newblock Translation-invariant monotone systems, and a global convergence
  result for enzymatic futile cycles.
\newblock {\em Nonlinear Analysis: Real World Applications}, 9(1):128 -- 140,
  2008.

\bibitem{Angeli2012}
D.~Angeli and E.D. Sontag.
\newblock Remarks on the invalidation of biological models using monotone
  systems theory.
\newblock In {\em Proceedings of the 51st IEEE Conference on Decision and
  Control (CDC)}, pages 2989--2994, 2012.

\bibitem{Birkhoff1957}
G.~Birkhoff.
\newblock Extensions of {J}entzsch's theorem.
\newblock {\em Transactions of the American Mathematical Society}, 85(1):pp.
  219--227, 1957.

\bibitem{Bonnabel2011}
S.~Bonnabel, A.~Astolfi, and R.~Sepulchre.
\newblock Contraction and observer design on cones.
\newblock In {\em Proceedings of 50th IEEE Conference on Decision and Control
  and European Control Conference}, pages 7147--7151, 2011.

\bibitem{boothby2003}
W.M. Boothby.
\newblock {\em An Introduction to Differentiable Manifolds and Riemannian
  Geometry, Revised}.
\newblock Pure and Applied Mathematics Series. Acad. Press, 2003.

\bibitem{Bushell1973}
P.J. Bushell.
\newblock Hilbert's metric and positive contraction mappings in a {B}anach
  space.
\newblock {\em Archive for Rational Mechanics and Analysis}, 52(4):330--338,
  1973.

\bibitem{Crouch1987}
P.E. Crouch and A.J. van~der Schaft.
\newblock {\em Variational and Hamiltonian control systems}.
\newblock Lecture notes in control and information sciences. Springer, 1987.

\bibitem{Dancer1998}
E.N. Dancer.
\newblock Some remarks on a boundedness assumption for monotone dynamical
  systems.
\newblock {\em Proceedings of the American Mathematical Society}, 126(3):pp.
  801--807, 1998.

\bibitem{DeLeenheer2001}
P.~De~Leenheer and D.~Aeyels.
\newblock Stabilization of positive linear systems.
\newblock {\em Systems \& Control Letters}, 44(4):259 -- 271, 2001.

\bibitem{DeLeenheer2004}
P.~De~Leenheer, D.~Angeli, and E.D. Sontag.
\newblock A tutorial on monotone systems - with an application to chemical
  reaction networks.
\newblock In {\em Proceedigns f the 16th International Symposium Mathematical
  Theory of Networks and Systems}, 2004.

\bibitem{DeLeenheer2007}
P.~De~Leenheer, D.~Angeli, and E.D. Sontag.
\newblock Monotone chemical reaction networks.
\newblock {\em Journal of Mathematical Chemistry}, 41(3):295--314, 2007.

\bibitem{Enciso2005}
G.~Enciso and E.D. Sontag.
\newblock Monotone systems under positive feedback: multistability and a
  reduction theorem.
\newblock {\em Systems \& Control Letters}, 54(2):159 -- 168, 2005.

\bibitem{Farina2000}
L.~Farina and S.~Rinaldi.
\newblock {\em Positive linear systems: theory and applications}.
\newblock Pure and applied mathematics (John Wiley \& Sons). Wiley, 2000.

\bibitem{Forni2014}
F.~Forni and R.~Sepulchre.
\newblock A differential {L}yapunov framework for contraction analysis.
\newblock {\em IEEE Transactions on Automatic Control}, 59(3):614--628, 2014.

\bibitem{Gedeon2010}
T.~Gedeon.
\newblock Oscillations in monotone systems with a negative feedback.
\newblock {\em SIAM Journal on Applied Dynamical Systems}, 9(1):84--112, 2010.

\bibitem{Gedeon2007}
T.~Gedeon and E.D. Sontag.
\newblock Oscillations in multi-stable monotone systems with slowly varying
  feedback.
\newblock {\em Journal of Differential Equations}, 239(2):273--295, Aug 2007.

\bibitem{Hardin2007}
H.M. Hardin and J.H. van Schuppen.
\newblock Observers for linear positive systems.
\newblock {\em Linear Algebra and its Applications}, 425(2--3):571 -- 607,
  2007.

\bibitem{Hirsch1988}
M.W. Hirsch.
\newblock Stability and convergence in strongly monotone dynamical systems.
\newblock {\em Journal f{\"u}r die reine und angewandte Mathematik}, 383:1--53,
  1988.

\bibitem{Hirsch1995}
M.W. Hirsch.
\newblock Fixed points of monotone maps.
\newblock {\em Journal of Differential Equations}, 123(1):171 -- 179, 1995.

\bibitem{Hirsch2003}
M.W. Hirsch and H.L. Smith.
\newblock Competitive and cooperative systems: A mini-review.
\newblock In L.~Benvenuti, A.~Santis, and L.~Farina, editors, {\em Positive
  Systems}, volume 294 of {\em Lecture Notes in Control and Information
  Science}, pages 183--190. Springer Berlin Heidelberg, 2003.

\bibitem{Knorn2009}
F.~Knorn, O.~Mason, and R~Shorten.
\newblock On linear co-positive {L}yapunov functions for sets of linear
  positive systems.
\newblock {\em Automatica}, 45(8):1943 -- 1947, 2009.

\bibitem{Pratt1982}
E.~Kohlberg and J.W. Pratt.
\newblock The contraction mapping approach to the {P}erron-{F}robenius theory:
  Why {H}ilbert's metric?
\newblock {\em Mathematics of Operations Research}, 7(2):pp. 198--210, 1982.

\bibitem{Lawson1989}
J.D. Lawson.
\newblock Ordered manifolds, invariant cone fields, and semigroups.
\newblock {\em Forum mathematicum}, 1(3):273--308, 1989.

\bibitem{Lemmens2012}
B.~Lemmens and R.D. Nussbaum.
\newblock {\em Nonlinear Perron-Frobenius theory}.
\newblock Cambridge tracts in mathematics, 189. Cambridge University Press,
  2012.

\bibitem{Luenberger1979}
D.G. Luenberger.
\newblock {\em Introduction to Dynamic Systems: Theory, Models, and
  Applications}.
\newblock Wiley, 1 edition, 1979.

\bibitem{Moreau2004}
L.~Moreau.
\newblock Stability of continuous-time distributed consensus algorithms.
\newblock In {\em 43rd IEEE Conference on Decision and Control}, volume~4,
  pages 3998 -- 4003, 2004.

\bibitem{Muratori1991}
S.~Muratori and S.~Rinaldi.
\newblock Excitability, stability, and sign of equilibria in positive linear
  systems.
\newblock {\em Systems \& Control Letters}, 16(1):59 -- 63, 1991.

\bibitem{Neeb1993}
K.H. Neeb.
\newblock Ordered symmetric spaces.
\newblock In {\em Proceedings of the Winter School "Geometry and Physics"},
  pages 21--26. Circolo Matematico di Palermo, 1993.

\bibitem{Nussbaum1994}
R.D. Nussbaum.
\newblock Finsler structures for the part metric and {H}ilbert's projective
  metric and applications to ordinary differential equations.
\newblock {\em Differential Integral Equations}, 7(5-6):1649--1707, 1994.

\bibitem{Olfati-Saber2007}
R.~Olfati-Saber, J.A. Fax, and R.M. Murray.
\newblock Consensus and cooperation in networked multi-agent systems.
\newblock {\em Proceedings of the IEEE}, 95(1):215--233, Jan 2007.

\bibitem{Piccardi2002}
C.~Piccardi and S.~Rinaldi.
\newblock Remarks on excitability, stability and sign of equilibria in
  cooperative systems.
\newblock {\em Systems \& Control Letters}, 46(3):153 -- 163, 2002.

\bibitem{Rantzer2012}
A.~Rantzer.
\newblock Distributed control of positive systems.
\newblock {\em ArXiv e-prints}, 2012.

\bibitem{Rockafellar2004}
R.J.-B. Wets.~M. Rockafellar, R.T.~Wets.
\newblock {\em Variational Analysis}.
\newblock Springer, 2004.

\bibitem{Roszak2009}
B.~Roszak and E.J. Davison.
\newblock Necessary and sufficient conditions for stabilizability of positive
  {LTI} systems.
\newblock {\em Systems \& Control Letters}, 58(7):474 -- 481, 2009.

\bibitem{Schaefer1971}
H.H. Schaefer.
\newblock {\em Topological vector spaces}.
\newblock Graduate Texts in Mathematics. Springer, 3rd edition, 1971.

\bibitem{Sepulchre2010a}
R.~Sepulchre.
\newblock Consensus on nonlinear spaces.
\newblock In {\em 8th IFAC Symposium on Nonlinear Control Systems}, 2010.

\bibitem{Sepulchre2010}
R.~Sepulchre, A.~Sarlette, and P.~Rouchon.
\newblock Consensus in non-commutative spaces.
\newblock In {\em Proceedings of the 49th IEEE Conference on Decision and
  Control, CDC 2010}, pages 6596--6601. IEEE, 2010.

\bibitem{Smith1995}
H.L. Smith.
\newblock {\em Monotone Dynamical Systems: An Introduction to the Theory of
  Competitive and Cooperative Systems}, volume~41 of {\em Mathematical Surveys
  and Monographs}.
\newblock American Mathematical Society, 1995.

\bibitem{Sontag1998}
E.D. Sontag.
\newblock {\em Mathematical Control Theory: Deterministic Finite Dimensional
  Systems}.
\newblock Texts in Applied Mathematics. Springer, 1998.

\bibitem{Sontag2007}
E.D. Sontag.
\newblock Monotone and near-monotone biochemical networks.
\newblock {\em Systems and Synthetic Biology}, 1(2):59--87, 2007.

\bibitem{Strogatz1994}
S.H. Strogatz.
\newblock {\em Nonlinear Dynamics And Chaos}.
\newblock Westview Press, 1994.

\bibitem{Trotta2012}
L.~Trotta, E.~Bullinger, and R.~Sepulchre.
\newblock Global analysis of dynamical decision-making models through local
  computation around the hidden saddle.
\newblock {\em PLoS ONE}, 7(3):e33110, 03 2012.

\bibitem{Willems1976}
J.C. Willems.
\newblock Lyapunov functions for diagonally dominant systems.
\newblock {\em Automatica}, 12(5):519--523, 1976.

\bibitem{Zhai2011}
C.B. Zhai and Z.D. Liang.
\newblock Hilbert's projective metric and the norm on a {B}anach space.
\newblock {\em Journal of Mathematical Research with Applications},
  31(1):91--99, 2011.

\end{thebibliography}

\end{document}